\documentclass[11pt,english,oneside]{amsart}
\usepackage{mathpazo}
\usepackage[scaled=0.92]{helvet}
\usepackage{courier}
\usepackage[T1]{fontenc}
\usepackage[latin9]{inputenc}
\usepackage[a4paper]{geometry}
\usepackage{amsthm}
\usepackage{amssymb}
\usepackage{setspace}
\usepackage{numprint} 
\npthousandsep{,}
\usepackage[authoryear]{natbib}
\usepackage{float}
\usepackage{graphicx}
\usepackage{xcolor}
\usepackage[normalem]{ulem}
\useunder{\uline}{\ul}{}
\usepackage[ruled]{algorithm2e}
\SetKwInput{KwInput}{Input}
\SetKwInput{KwOutput}{Output}

\makeatletter
\numberwithin{equation}{section}

\usepackage{chngcntr}

\newtheorem{theorem}{Theorem} 
\theoremstyle{plain}

\newtheorem{assumption}{Assumption} 

\numberwithin{equation}{section}
\renewcommand{\v}[1]{\boldsymbol{#1}}
\newcommand{\ci}{\mathrm{i}}

\makeatother

\usepackage{babel}
\begin{document}
\title[Multivariate Spectral Subsampling MCMC]{Spectral Subsampling MCMC for Stationary Multivariate Time Series with Applications to Vector ARTFIMA Processes}
\author{Mattias Villani, Matias Quiroz, Robert Kohn, and Robert Salomone}
\thanks{
Villani: \textit{Department of Statistics, Stockholm University, SE-106 91 Stockholm, Sweden} and 
\textit{Department of Computer and Information Science, Link{\"o}ping University}.
\textit{E-mail: mattias.villani@stat.su.se}.
Quiroz: \textit{School of Mathematical and Physical Sciences, University of Technology Sydney}.
Kohn: \textit{School of Business, University of New South Wales}.
Salomone:  \textit{Centre for Data Science, Queensland University of Technology}.
}

\begin{abstract}
Spectral subsampling MCMC was recently proposed to speed up Markov chain Monte Carlo (MCMC) for long stationary univariate time series by subsampling periodogram observations in the frequency domain. This article extends the approach to multivariate time series using a multivariate generalisation of the Whittle likelihood. To assess the computational gains from spectral subsampling in challenging problems, a multivariate generalisation of the autoregressive tempered fractionally integrated moving average model (ARTFIMA) is introduced and some of its properties derived. Bayesian inference based on the Whittle likelihood is demonstrated to be a fast and accurate alternative to the exact time domain likelihood. Spectral subsampling is shown to provide up to two orders of magnitude additional speed-up, while retaining MCMC sampling efficiency and accuracy, compared to spectral methods using the full dataset. 
\end{abstract}

\maketitle

\section{Introduction}
Recent technological developments in sensors, data storage and computing power make it possible to collect high frequency time series data at low cost; some examples are financial transaction data \citep{mykland2012econometrics}, neuroimaging data with high temporal resolution \citep{chen2019analysis}, sensor data from robots \citep{deisenroth2013gaussian} or meteorological weather stations, and GPS and smart card data used in transportation \citep{welch2019big}.

However, statistical analysis of time series with tens of thousands, hundreds of thousands, or even millions of data points is computationally challenging, especially when inferences are obtained with iterative methods such as Markov chain Monte Carlo (MCMC) simulation or stochastic optimisation algorithms, where the likelihood is evaluated a large number of times. It is therefore common to only use a portion of the data for inference, for example only the most recent observations, or by systematically selecting every $k$th data point over the study period. Such downsampling wastes valuable data, gives less precise inferences, and is not even an option when predictions are required at the original sampling frequency.

\citet{salomone2019spectral} propose \emph{spectral subsampling MCMC} to accelerate MCMC for long stationary univariate time series. They use the asymptotically motivated \textit{Whittle likelihood} \citep{whittle1953analysis} to approximate the likelihood of a stationary time series. The Whittle likelihood is based on the discrete Fourier transform with the important property of transforming a time series with dependent observations to \emph{asymptotically independent} periodogram observations in the frequency domain. The key insight in \citet{salomone2019spectral} is that such independence makes it possible to extend subsampling MCMC approaches for independent data \citep{quiroz2018delayed, quiroz2019speeding, quiroz2020block, dang2019hamiltonian} to univariate stationary time series by systematic subsampling of periodogram observations. 

All of the applications mentioned above are naturally analysed in a multivariate setting: financial portfolios consisting of many assets, neuroimaging data simultaneously measured at multiple brain locations, meteorological data collected at several spatial locations, and so on. The multivariate aspect naturally leads to even more demanding computations in a high frequency setting. 

The main contribution in this article is extending the spectral subsampling MCMC methodology to stationary \emph{vector-valued} time series by using a multivariate version of the Whittle likelihood based on the asymptotic properties of the matrix-valued periodogram. The proposed multivariate spectral subsampling MCMC algorithm is evaluated on three challenging large-scale multivariate time series applications from meteorology, hydrology and environmental science. 

The multivariate Whittle likelihood is of general interest for large-scale likelihood and Bayesian inference beyond subsampling MCMC. Even without subsampling, the Whittle likelihood is substantially more scalable to large data as it sidesteps the costly matrix inversions needed for the exact likelihood in the time domain. But the independence of periodogram observations also makes the multivariate Whittle likelihood directly useful for computing the unbiased gradient estimates from random subsets/batches of frequencies needed for large-scale variational inference \citep{tran2017variational} or maximum likelihood estimates from stochastic gradient descent algorithms \citep{goodfellow2016deep}. An additional contribution of our article is that our empirical results clearly demonstrate that the posterior based on the multivariate Whittle likelihood gives an excellent approximation to the exact time domain posterior, even in complex models such as the VARMA model.

A challenge for subsampling MCMC methods is keeping the variance of the likelihood estimator small enough for the MCMC chain to mix well. Previous literature is therefore typically restricted to applications using models with a very small number of parameters, or models with a moderate number of parameters with a simple structure, such as logistic regression. To test our multivariate spectral subsampling MCMC methodology on a challenging set of problems, we introduce a new multivariate model with semi-long range dependence that extends the autoregressive tempered fractionally integrated moving average (ARTFIMA, \citet{Sabzikar2019}) to the multivariate setting. Several properties of the vector ARTFIMA model are derived, including the spectral density matrix needed for the Whittle likelihood. For all examples considered, the tempered fractional differencing is shown to improve upon the standard multivariate vector autoregressive integrated moving average (VARIMA) model in terms of the Bayesian information criterion. Our empirical examples show that spectral subsampling works well and gives a very large speed-up for time series that are long enough to make the variance reducing control variates effective. We also highlight the limitation of subsampling MCMC by showing empirically that spectral subsampling MCMC can get stuck when estimating complex multivariate models on shorter time series.

The rest of the article is organised as follows. Section \ref{sec:multivar_spectral_mcmc} presents the Whittle likelihood for multivariate stationary time series. Section \ref{sec:subsampling_mcmc} outlines the subsampling MCMC methodology. Section \ref{sec:Models} presents the models considered in our applications and establishes properties needed for the implementation of spectral methods. Section \ref{sec:Applications} demonstrates the methodology and evaluates the efficiency of the proposed spectral subsampling algorithm on real data. Section \ref{sec:univariate} compares our proposed vector ARTIFIMA to the alternative of separately estimating univariate ARTFIMA models for each series \citep{Sabzikar2019}, both in terms of estimated parameters and in forecasting. Section \ref{sec:Conclusions} concludes and discusses future research. Appendix \ref{app:proofs} derives some properties of the vector ARTFIMA model and Appendix \ref{app:additionalresults} contains additional empirical results.

\section{The Whittle likelihood for multivariate time series}\label{sec:multivar_spectral_mcmc}
Let $\mathbf{X}_t \in \mathbb{R}^r$ be an $r$-variate zero mean stationary time series with absolutely summable autocovariance matrix function
\begin{equation}
    \v\gamma_{\mathbf{X}}(\tau) = \mathrm{Cov}(\mathbf{X}_t,\mathbf{X}_{t-\tau}),\text{ for } \tau \in \mathbb{Z},
\end{equation}
where $\mathbb{Z}$ is the set of integers. 
The spectral density matrix is
\begin{equation}
    f_{\mathbf{X}}(\omega) = \frac{1}{2\pi}\sum_{\tau=-\infty}^\infty 
        \v\gamma_{\mathbf{X}}(\tau)\exp(-\ci \omega \tau), \text{ for } \omega \in (-\pi,\pi],
\end{equation}
with the diagonal elements being the usual spectral density for each univariate time series and the off-diagonal
elements are the cross-spectral densities 
\begin{equation}
    f_{jk}(\omega) = \frac{1}{2\pi}\sum_{\tau=-\infty}^\infty \gamma_{jk}(\tau)\exp(-\ci \omega \tau), \text{ for } \omega \in (-\pi,\pi].
\end{equation}   
Since the elements of $\v\gamma_{\mathbf{X}}(\tau)$ are real, $f_{\mathbf{X}}(\omega)$ is Hermitian, i.e.\ $f_{\mathbf{X}}(\omega)^H=f_{\mathbf{X}}(\omega)$, where $\mathbf{A}^H = (\overline{\mathbf{A}})^\top$ is the conjugate transpose of a matrix and $\overline{\mathbf{A}}$ is the matrix of complex conjugates of the elements of $\mathbf{A}$. \citet{brillinger2001time}[Theorem 2.5.1] proves that $f_{\mathbf{X}}(\omega)$ is also non-negative definite.

The discrete Fourier transform (DFT) of the multivariate time series $\{\mathbf{X}_t\}_{t=0}^{T-1}$ is
\begin{equation}
    J_T(\omega_k) = \sum_{t=0}^{T-1} \mathbf{X}_t \exp(-\ci \omega_k t),
\end{equation}
for $\omega_k \in \Omega_T = \{2\pi k/T \text{ for } k=-\lceil T/2\rceil+1,\ldots,\lfloor T/2 \rfloor\}$, the set of Fourier frequencies.
Let $\v X \sim \mathrm{CN}(\v \mu, \v \Sigma)$ denote that the $r$-dimensional complex-valued vector $\v X$ follows the multivariate complex normal distribution \citep[Ch. 4.2]{brillinger2001time}, i.e.\ that
\begin{align*}
    \begin{pmatrix}
        \mathrm{Re}~\v X \\
        \mathrm{Im}~\v X  
    \end{pmatrix}
    \sim \mathrm{N}_{2r} \Bigg[ 
    \begin{pmatrix}
        \mathrm{Re}~\v \mu \\
        \mathrm{Im}~\v \mu 
    \end{pmatrix},
    \frac{1}{2}\begin{pmatrix}
        \mathrm{Re}~\v \Sigma & -\mathrm{Im}~\v \Sigma\\
        \mathrm{Im}~\v \Sigma & \hphantom{-}\mathrm{Re}~\v \Sigma
    \end{pmatrix}
    \Bigg],
\end{align*}
where $\v \Sigma$ is a Hermitian positive definite covariance matrix.
The $J_T(\omega_k)$ are asymptotically independent over the different frequencies and 
\citep[Theorem 4.4.1]{brillinger2001time}
\begin{equation}
    \frac{1}{\sqrt{T}}J_T(\omega_k) \sim \mathrm{CN}(0,2\pi f_{\mathbf{X}}(\omega_k)) \text{ as } T\rightarrow \infty,
\end{equation}
except at $\omega_k=0$ and $\omega_k=\pi$, where instead $(1/\sqrt{T})J_T(\omega) \sim \mathrm{N}_r(0,2\pi f_{\mathbf{X}}(\omega_k))$ as $T\rightarrow \infty$. Following \citet{brillinger2001time}, we will assume that $f_{\mathbf{X}}(\omega)$ is non-singular.

The \emph{periodogram} ordinate at frequency $\omega$ is defined as 
$$I_T(\omega) = (2\pi T)^{-1}J_T(\omega)J_T(\omega)^H.$$ 
The periodogram ordinates are therefore asymptotically independent complex Wishart distributed with one degree of freedom $I_T(\omega) \sim 
\mathrm{CW}_r(1,f_{\mathbf{X}}(\omega))$ \citep[Theorem 7.2.4]{brillinger2001time}, except for the frequencies $\omega_k=0$ and $\omega_k=\pi$ where the $I_T(\cdot)$ instead follow a (real) Wishart distribution. The periodogram ordinates $I_T(\omega)$ are singular matrices for $r>1$. The density function of this singular Wishart distribution is derived for the real case by \cite{uhlig1994singular} with respect to the Hausdorff measure, and by \citet{srivastava2003singular} with respect the Lebesgue measure for the functionally independent elements of the matrix. The density for the  complex singular Wishart distribution $\mathbf{W}\sim \mathrm{CW}_r(\nu,\Sigma)$ for $\nu < r$ over the space of $r \times r$ positive semidefinite matrices of rank $\nu$ is derived in \citet{ratnarajah2005complex}[Theorem 3] as
\begin{equation*}
    p(\mathbf{W}|\nu,\Sigma)= \frac{\pi^{\nu(\nu-r)}}{\Gamma_\nu(\nu)}\big(\prod_{j=1}^\nu \ell_j\big)|\Sigma|^{-\nu}\exp(-\mathrm{tr}~\Sigma^{-1}\mathbf{W}),
\end{equation*}
where $\Gamma_\nu(\nu) = \pi^{\nu(\nu-1)/2}\prod_{k=1}^\nu\Gamma(\nu-k+1)$, $\ell_j$ is the $j$th eigenvalue in the reduced spectral decomposition $\v W=\v E_1 \v \Lambda \v E_1^H$, where $\v E_1$ is an $r\times \nu$ complex orthogonal matrix and $\v \Lambda = \mathrm{Diag}(\ell_1,\ldots,\ell_\nu)$. 

Consider now inference for a parametric model with parameter vector $\v \theta$ and spectral density matrix $f_{\mathbf{X},\v \theta}(\omega_k)$. The Whittle log-likelihood exploits the asymptotic result for the periodogram and is defined using the complex singular Wishart distribution as
\begin{equation}\label{eq:multiwhittle}
    \ell_\mathcal{W}(\v \theta) = - \sum_{\omega_k\in \tilde \Omega_T} \left( \log | f_{\mathbf{X},\v \theta}(\omega_k)| 
     + \mathrm{tr}\left[f_{\mathbf{X},\v \theta}(\omega_k)^{-1}I_T(\omega)\right] \right),
\end{equation}
where $\tilde \Omega_T$ is the set of Fourier frequencies with the omission of $\omega_k=0$ and $\omega_k=\pi$. The term for $\omega_k=0$ is not included when the time series is demeaned since then $J_T(0)=0$ by construction; the term for $\omega_k=\pi$ is removed for simplicity since it has a different distribution than the other frequencies and its influence is negligible asymptotically.
Note that since $f_{\mathbf{X},\v \theta}(\omega_k)$ is Hermitian for an absolutely summable stationary process \citep{brillinger2001time},
both $| f_{\mathbf{X},\v \theta}(\omega_k)|$ and
\[\mathrm{tr}\left[f_{\mathbf{X},\v \theta}(\omega_k)^{-1}I_T(\omega)\right] = I_T(\omega)^H f_{\mathbf{X},\v \theta}(\omega_k)^{-1}I_T(\omega),\] are real-valued.

\section{Subsampling Markov chain Monte Carlo}\label{sec:subsampling_mcmc}
Subsampling MCMC uses the framework of pseudo-marginal MCMC \citep{Andrieu2009}, in which an estimator of the likelihood is used within a Metropolis-Hastings algorithm. This section gives a briefly reviews subsampling MCMC; see \cite{quiroz2019speeding} and \cite{salomone2019spectral} for details on how pseudo-marginal methods are used for subsampling problems.

Let $\pi({\v \theta}) \propto L_n({\v \theta})p({\v \theta})$ denote the posterior distribution of the model parameter $\v \theta$ from a sample of $n$ observations with likelihood function $L_n({\v \theta})$, and prior distribution $p({\v \theta})$. Metropolis-Hastings MCMC algorithms sample iteratively from $\pi({\v \theta})$ by proposing a parameter vector ${\v \theta}^{(j)}$ at the $j$th iteration from the proposal distribution $g(\v\theta^{(j)} \vert \v\theta^{(j-1)})$ and accepting the draw with probability
\begin{equation}\label{eq:MHaccProb}
\min\Bigg\{1,
    \frac{L_n(\v\theta^{(j)})p(\v\theta^{(j)})}{L_n(\v\theta^{(j-1)})p(\v\theta^{(j-1)})}
    \cdot
    \frac{g(\v\theta^{(j-1)} \vert \v\theta^{(j)})}{g(\v\theta^{(j)} \vert \v\theta^{(j-1)})}
    \Bigg\}.
\end{equation}

The cost of computing the likelihood $L_n(\v \theta)$ in the acceptance probability  \eqref{eq:MHaccProb} is a major concern when $n$ is large. \citet{quiroz2019speeding} propose speeding up MCMC for large $n$ by replacing $L_n(\v \theta)$ with an estimate $\widehat L(\v \theta,\v u)$ based on a small random subsample of $m\ll n$ observations, where $\v u=(u_1,...,u_m)$ indexes the selected observations. Their algorithm samples $\v\theta$ and $\v u$ jointly from an extended target distribution $\tilde \pi (\v\theta,\v u) \propto \widehat L(\v \theta,\v u)p(\v \theta)p(\v u)$. \citet{Andrieu2009} show that such pseudo-marginal MCMC algorithms sample from the full-data posterior $\pi({\v \theta})$ if $\widehat L(\v \theta,\v u)$ is an unbiased and almost surely positive estimator of $L_n(\v \theta)$, where the unbiasedness condition is $\mathrm{E}_{\v u}\widehat L(\v \theta,\v u) = \int \widehat L(\v \theta,\v u)p(\v u)d\v u = L_n(\v \theta)$. \citet{quiroz2019speeding} use an unbiased estimator of the log-likelihood $\widehat \ell(\v\theta,\v u)$ and then debias $\exp(\widehat \ell(\v\theta,\v u))$ to estimate the full-data likelihood. This debiasing approach does not remove all bias and the marginal distribution of $\v \theta$ from their pseudo-marginal sampler is the slightly perturbed posterior
\begin{align}
    \label{eq:perturbed_posterior}
    \eta(\v \theta) & = \frac{\left(\int \widehat{L}(\v\theta, \v u)p(\v u) d\v u\right) p(\v \theta)}{\int \left(\int \widehat{L}(\v\theta, \v u)p(\v u) d\v u\right) p(\v \theta) d\v \theta}.
\end{align}
The perturbation error of $\eta(\v \theta)$ in \eqref{eq:perturbed_posterior} is within $O(n^{-1}m^{-2})$ distance in total variation norm of the true posterior $\pi(\v \theta)$, and the applications in \citet{quiroz2019speeding, quiroz2020block} and  \citet{dang2019hamiltonian} show negligible bias. \cite{salomone2019spectral} show that these results extend to subsampling for the Whittle likelihood, with the true posterior in this case based on the Whittle likelihood for the full data. 

To apply subsampling MCMC, the log-likelihood needs to decompose as a sum $\ell(\v \theta) = \sum_{k=1}^n \ell_k(\v\theta)$; either by assuming independent data or by using the Whittle likelihood in the frequency domain for temporally dependent data as in \eqref{eq:multiwhittle}. Estimating the log-likelihood is analogous to the problem of estimating a population total in survey sampling \citep{quiroz2018sankhya}. 

It is by now well-known that subsampling MCMC requires a likelihood estimator with small variance, otherwise the sampler tends to get stuck \citep{quiroz2019speeding,quiroz2018sankhya}. \cite{quiroz2019speeding} propose using the \textit{difference estimator} 
\begin{equation*}
 \widehat \ell(\v\theta,\v u) = q(\v \theta) + \frac{n}{m}\sum_{i=1}^m \left( \ell_{u_i}(\v \theta) - q_{u_i}(\v \theta)\right), \quad \text{with } q(\v \theta) = \sum_{k=1}^n q_k(\v\theta),
\end{equation*}
with control variates $q_k(\v\theta)$ to reduce the variance. The second term in the first equation is an unbiased estimate of $\ell(\v \theta)-q(\v \theta)$ and is efficient if $q_k(\v \theta) \approx \ell_k(\v \theta)$ for $k=1, \dots, n$. The control variates homogenise the log-density terms in the estimator so that the observations can be sampled by simple random sampling \citep{quiroz2019speeding}. \cite{bardenet2017markov} propose setting $q_k(\v \theta)$ equal to a second order Taylor expansion around $\v \theta^\star$ (e.g.\ the posterior mode). An important property of this control variate is that the $q(\v \theta)$ term can be computed in $O(1)$ time. 

\begin{figure}
    \includegraphics[width=0.55\linewidth]{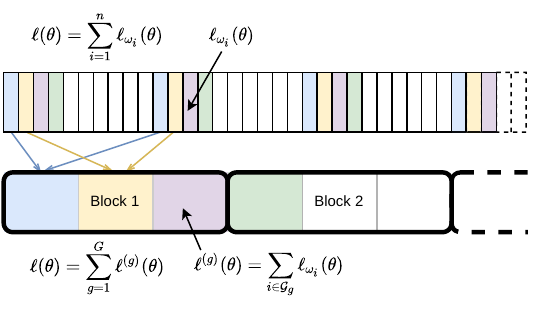}
    \caption{Explaining grouping and blocking in spectral subsampling MCMC. The periodogram observations and their corresponding log-density terms $\ell_{\omega_i}(\theta)$ (top row) are systematically divided into groups (indicated by colors) and summed $\ell^{(g)}(\theta)$ (bottom row, $\mathcal{G}_g$ denotes the indices belonging to group $g$), so that each group spans over the whole frequency domain. The groups are then divided into blocks and only the group subsampling indicators within a subset of the blocks are updated at each MCMC iteration.}\label{fig:ExplainGroupAndBlock}
\end{figure}

The control variate in \cite{bardenet2017markov} works well when the log-density $\ell_{k}(\v \theta)$ of each observation is approximately quadratic in the neighbourhood of $\v \theta^\star$ explored by the MCMC. There is, however, no guarantee that the individual log-densities $\ell_{k}(\v \theta)$ are close to quadratic in complex models, particularly in large parameter spaces where the posterior may not be highly concentrated around $\v \theta^\star$. \cite{salomone2019spectral} therefore propose a grouped quadratic control variate, in which observations are divided into $G$ groups and the log-likelihood contribution for each group is approximated by a quadratic function. The idea is that the Bernstein-von Mises theorem (asymptotic normality of the posterior) suggests an approximately quadratic log-likelihood for the
group, given that the number of observations in the group is large enough. $G$ should be chosen such that there are enough observations in each group for the approximation to be accurate. Note that the sampling units in this approach are the groups, and hence there are $G$ group log-likelihoods to subsample.  For the Whittle likelihood, the order of the dataset corresponds to the order of the frequencies in the DFT (from low to high). The groups should then be formed via systematic sampling, so that each group contains periodogram ordinates over the whole frequency spectrum; Figure \ref{fig:ExplainGroupAndBlock} graphically illustrates this and the block setup which is now described.

\begin{algorithm}
\SetAlgoLined
\KwInput{data \textbf{y}, unbiased log-likelihood estimator $\widehat{\ell}(\v\theta,\v u)$ based on Taylor control variates, variance estimator $\widehat{\sigma}^2_{\widehat{\ell}}(\v\theta,\v u)$ of $\widehat{\ell}(\v\theta,\v u)$, prior density $p(\boldsymbol{\theta})$, initial value $\boldsymbol{\theta}^{(0)}$, initial subsample $\v u^{(0)}=(\v u_1^{(0)}, \dots, \v u_K^{(0)})$ divided into $K$ blocks, proposal covariance $\Omega$ for $\v \theta$.} 
\BlankLine
define $\hat p(\mathbf{y} | \boldsymbol{\theta}, \v u) = \exp\left(\widehat{\ell}(\v\theta,\v u) - \frac{1}{2} \widehat{\sigma}^2_{\widehat{\ell}}(\v\theta,\v u) \right)$ \\
\For{$j = 1$ \KwTo $N$}{
    \BlankLine
    \textcolor{blue}{\textbackslash\textbackslash\hspace{0.1cm} propose the subsample $\v u$}  \\
    sample $k$ uniformly from $\{1, \dots, K\}$, generate $\v u_k^{\prime} \sim p(\v u_k)$ and propose $$\v u^\prime=(\v u_1^{(j-1)}, \dots, \v u_k^{\prime}, \dots, \v u_K^{(j-1)})$$ \\
    \textcolor{blue}{\textbackslash\textbackslash\hspace{0.1cm}propose the parameter $\v\theta$} \\
    generate $\v\theta^{\prime} \sim N(\v\theta^{(j-1)}, \Omega)$ \\
    set $\v u^{(j)}, \v\theta^{(j)}   \leftarrow \v u^{\prime}, \v\theta^{\prime}$ with probability \\
    \hspace{2cm}$\alpha= \min \Big( 1, \frac{\hat p(\mathbf{y} | \boldsymbol{\theta}^{\prime}, \v u^{\prime})p(\v\theta^{\prime})}{\hat p(\mathbf{y} | \boldsymbol{\theta}^{(j-1)}, \v u^{(j-1)})p(\v\theta^{(j-1)})}   \Big)$ \\
    else set $\v u^{(j)}, \v\theta^{(j)} \leftarrow \v u^{(j-1)}, \v\theta^{(j-1)}$
} 
\BlankLine
\KwOutput{autocorrelated random draws $\boldsymbol{\theta}^{(1)},\ldots,\boldsymbol{\theta}^{(N)}$ from $\eta(\v \theta)$ in \eqref{eq:perturbed_posterior}.}
\BlankLine
\caption{Subsampling MCMC with a random walk Metropolis proposal for $\v\theta$ and a block proposal for $\v u$. \label{alg:Subsampling_MCMC}}
\end{algorithm}

Pseudo-marginal methods can be made much more efficient by correlating the estimators used at the numerator and denominator of the Metropolis-Hastings acceptance ratio \citep{deligiannidis2018correlated, Tran2016block, quiroz2020block}. We use the  \cite{Tran2016block} approach and divide the random numbers $\v u$ (here the set of subsampled observations indices) into $K$ blocks. By updating the $\v u$ only within one of the blocks at each MCMC iteration, \cite{Tran2016block} show that the correlation between the logs of the estimators in the numerator and denominator is approximately $1-1/K$, and setting $K$ large makes the pseudo-marginal method less sensitive to the variance of the log of the likelihood estimator; see Figure \ref{fig:ExplainGroupAndBlock} again for the relation between the grouping described in the previous paragraph and the block setup. 

Algorithm \ref{alg:Subsampling_MCMC} sketches the implementation of the algorithm; see \cite{quiroz2019speeding} for details.

\section{The Vector ARTFIMA process}\label{sec:Models}

 \citet{salomone2019spectral} demonstrate that spectral subsampling MCMC can be successfully applied in univariate time series models with likelihood functions that are known to be non-Gaussian in the parameters due to local non-identification, such as ARMA and ARTFIMA. We explore the performance of spectral subsampling MCMC for multivariate versions of ARMA and ARTFIMA, which pose an even greater challenge since the number of parameters in multivariate models typically increase quadratically with the number of time series.

\subsection{Vector ARMA}\label{subsec:VARMA}
The vector autoregressive moving average VARMA($p,q$) model is
\begin{equation}
    \Phi(L)(\mathbf{Y}_t-\boldsymbol \mu) = \Theta(L)\boldsymbol{\varepsilon}_t,
\end{equation}
where $\{ \boldsymbol{\varepsilon}_t  \}_{t=1}^T$
is an iid sequence from $N(0,\Sigma_\varepsilon)$, 
$\Phi(L) = I_r - \Phi_1L - \cdots - \Phi_pL^p$ and $\Theta(L) = I_r + \Theta_1L + \cdots + \Theta_qL^q$ are the AR and MA 
lag polynomials, respectively. We assume that the usual conditions for stationarity and invertibility of the VARMA process hold:
\begin{assumption}\label{ass:VARMAroots}
    The matrix polynomials $\Phi(z)$ and $\Theta(z)$ share no common zeros and $|\Phi(z)|\neq 0$ and $|\Theta(z)| \neq 0$ for $|z| \leq 1$.
\end{assumption}
  
The spectral density matrix of the VARMA($p,q$) model is \citep[Ch. 11]{brockwell1991time}
\begin{equation}
    f_{\mathbf{Y}}(\omega) = \frac{1}{2\pi}\Phi^{-1}(e^{-i\omega})\Theta(e^{-i\omega}) \Sigma_{\varepsilon} \Theta^H(e^{-i\omega})  \Phi^{-H}(e^{-i\omega}),
\end{equation}
where $\v A^{-H}$ denotes the inverse of the conjugate transpose of the matrix $\v A$.

The Whittle likelihood is valid if the process is  stationary, i.e.\ if all roots of $\Phi(z)$ are outside of the unit circle. We therefore use the reparametrisation in \cite{ansley1986note} to map a set of unconstrained real-valued AR coefficient matrices to the set of stationary parameters. The same reparametrisation is also used on the MA parameter matrices to ensure invertibility. These reparametrisations are an additional source of non-linearity/non-Gaussianity; even a plain vector AR (VAR) model is no longer linear in the parameters.

\subsection{Vector ARTFIMA}\label{subsec:ARTFIMA}

\citet{Sabzikar2019} define the univariate $\mathrm{ARTFIMA}(p,d,\lambda,q)$ process for $Y_t$ as
\begin{equation}\label{eq:artfima}
    \phi(L)\Delta^{d,\lambda}(Y_{t}-\mu)=\vartheta(L)\varepsilon_{t},
\end{equation}
where $\{\varepsilon_t\}_{t \in \mathbb{Z}}$ is an iid sequence of zero mean random variables with variance $\sigma_{\varepsilon}^2$, 
$\phi(L)\equiv 1-\phi_{1}L-\cdots-\phi_{p}L^{q}$, and $\vartheta(L)\equiv 1+\vartheta_{1}L+\cdots+\vartheta_{q}L^{p}$,
are the autoregressive and moving average lag polynomials, and $L$ is the lag operator, i.e.\ $L^k Y_t \equiv Y_{t-k}$. 
The \emph{tempered fractional differencing operator} $\Delta^{d,\lambda}$, where $d\notin\mathbb{Z}$ is the \emph{fractional differencing} parameter and $\lambda \geq 0$ is the \emph{tempering} parameter, is defined by the generalised binomial theorem as
\begin{equation}
     \Delta^{d,\lambda}Y_t \equiv (1-e^{-\lambda}L)^d Y_t  =\sum_{j=0}^\infty b_j^{d,\lambda}Y_{t-j},   
\end{equation}
where 
\begin{equation*}
    b_j^{d,\lambda} \equiv (-1)^j \binom{d}{j}e^{-\lambda j} \text{ and } \binom{d}{j} = \frac{\Gamma(1+d)}{\Gamma(1+d-j)j!}.
\end{equation*}
We follow the convention in time series of not explicitly writing out the lag operator $L$ in differencing operators unless needed for clarity, i.e.,\ $\Delta^{d,\lambda}\equiv \Delta^{d,\lambda}(L)$.

To explain the role of the parameters $d$ and $\lambda$, note that for $\lambda=0$ and $d$ a non-negative integer, 
$\Delta^{d,\lambda}Y_t$ reduces to simple differencing of order $d$ and the ARTFIMA model in \eqref{eq:artfima} reduces to the autoregressive integrated moving average 
(ARIMA) process. For $\lambda=0$ and fractional $d$ we obtain the autoregressive fractionally integrated moving average (ARFIMA) model  \citep{Granger1980}. The ARFIMA process is stationary and invertible for $-0.5 < d < 0.5$, and has long-range or long-memory dependence with an autocovariance function dying off so slowly that it is not absolutely summable. The tempering parameter $\lambda>0$ in ARTFIMA allows for semi-long range dependence, i.e.\ ARFIMA-like long range dependence for a number of lags beyond which the autocovariances decay exponentially fast. \cite{Sabzikar2019} prove that the ARTFIMA process is stationary for any $d\notin\mathbb{Z}$ if $\lambda>0$, provided that the univariate case of Assumption \ref{ass:VARMAroots} holds.  

The univariate ARTFIMA process is now extended to the multivariate case, where $\mathbf{Y}_{t}$ is an $r$-dimensional vector-valued time series. We define the vector ARTFIMA (VARTFIMA) process as
\begin{equation}\label{eq:vartfima}
    \Phi(L)\Delta^{\mathbf{d},\v \lambda}(\mathbf{Y}_{t}-\v \mu)=\Theta(L) \v \varepsilon_{t},
\end{equation} 
where $\Delta^{\mathbf{d},\v \lambda}$ is the multivariate tempered fractional differencing operator defined by
\begin{equation}
    \Delta^{\mathbf{d},\v \lambda}\mathbf{Y}_{t} \equiv
    \Big(\Delta^{d_1,\v \lambda_1} Y_{1,t},\ldots, \Delta^{d_r,\v \lambda_r} Y_{r,t} \Big)^\top,
\end{equation}
hence allowing for different fractional differences and temporal differencing for the $r$ time series.

Theorem \ref{thm:stationary} generalises the spectral density result in \cite{Sabzikar2019} to the vector case. Its proof is in Appendix \ref{app:proofs}.
\begin{theorem}\label{thm:stationary}
    The multivariate ARTFIMA process is causal and stationary for all $d\notin \mathbb{Z}$ and all $\lambda>0$ if $|\Phi(z)|\neq 0$ for $|z| \leq 1$. The spectral density matrix is 
    \begin{equation*}
        f_{\mathbf{Y}}(\omega) = \frac{1}{2\pi} \Delta^{\mathbf{-d},\v \lambda}(e^{-i\omega}) \Phi^{-1}(e^{-i\omega})\Theta(e^{-i\omega}) \Sigma_{\varepsilon} \Theta^H(e^{-i\omega})  \Phi^{-H}(e^{-i\omega})  \Delta^{\mathbf{-d},\v \lambda}(e^{-i\omega})^H,
    \end{equation*}
    where $\Delta^{\mathbf{-d},\v \lambda}(z) = \mathrm{Diag}\big((1-e^{-\lambda_1}z)^{-d_1},\ldots,(1-e^{-\lambda_r }z)^{-d_r}\big)$, and $\v A^{-H}$ the inverse of the conjugate transpose of the matrix $\v A$.
\end{theorem}

\section{Applications}\label{sec:Applications}
We now illustrate the proposed subsampling MCMC methodology on three long multivariate time series datasets of varying dimensions. 

\subsection{Datasets}\label{subsec:Data}

The first dataset contains observations on mean water velocity from two measurement stations at St Clair River and Detroit River, located on opposite sides of Lake St Clair forming a connecting channel between Lake Huron and Lake Erie. The measurements are recorded every $12$th minute. The final dataset contains $\numprint{130,001}$ observations from Jan 3, 2016 at 00:00 hours until Dec 21, 2018 at 08:00 hours for each of the two locations. Missing observations are imputed using the \texttt{na.interp} function in the R package \texttt{forecast}. Figure \ref{fig:WaterVelocityData} plots the data.

The second dataset contains $\numprint{124,879}$ observations of Swedish temperatures at each of three airport locations (Arlanda, Bromma and Landvetter), giving a three-dimensional time series. The data are measured in an hourly scale for the time period February 1, 2008 until May 1, 2022, and are processed as follows. Each univariate series is preprocessed separately: missing observations are imputed with the value at a nearby location if available or otherwise imputed with the \texttt{na.interp} function in the R package \texttt{forecast}. Trend and seasonal components are then removed using the \texttt{mstl} function in the R package \texttt{stats} removing both a daily and an annual seasonal cycle. Figure \ref{fig:SwedTempData} plots the processed data. The last $878$ observations are used as a test set for evaluating prediction performance in Section \ref{sec:univariate} and are therefore not used for the posterior inference. The raw data before pre-processing are plotted in Figure \ref{fig:SwedTempRawData} in Appendix \ref{app:additionalresults}.

\begin{figure}
    \includegraphics[width=0.40\linewidth]{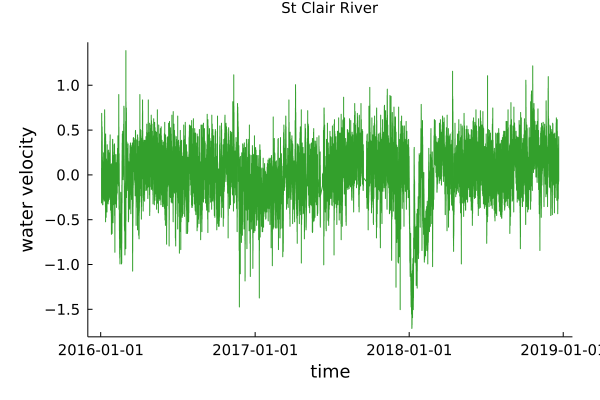}
    \includegraphics[width=0.40\linewidth]{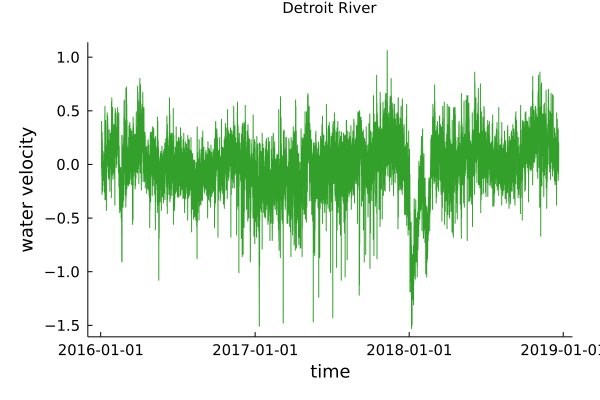}
    \caption{Water velocity data.}\label{fig:WaterVelocityData}
\end{figure}

\begin{figure}
    \includegraphics[width=0.32\linewidth]{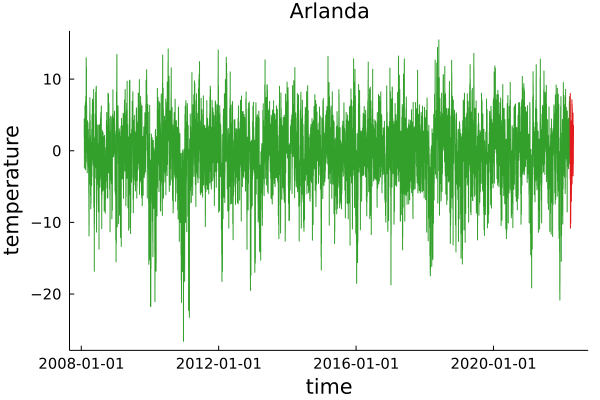}
    \includegraphics[width=0.32\linewidth]{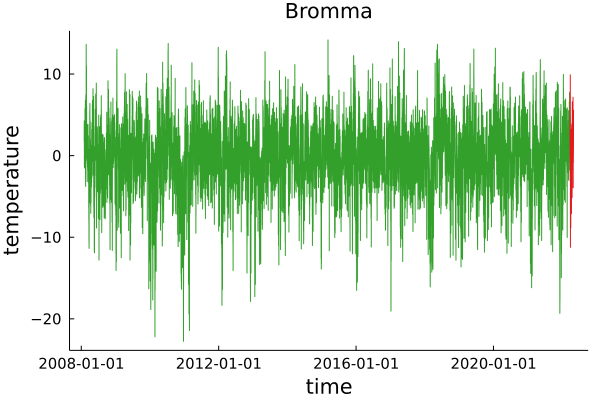}
    \includegraphics[width=0.32\linewidth]{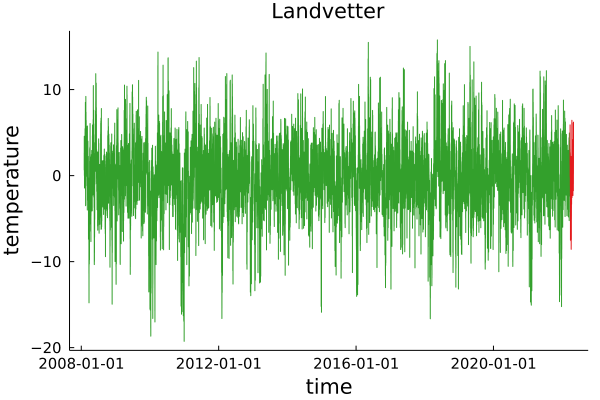}
    \caption{Swedish temperature data. The final $878$ time points in red are used for testing the forecast accuracy of the model.}\label{fig:SwedTempData}
\end{figure}

\begin{figure}
    \includegraphics[width=0.35\linewidth]{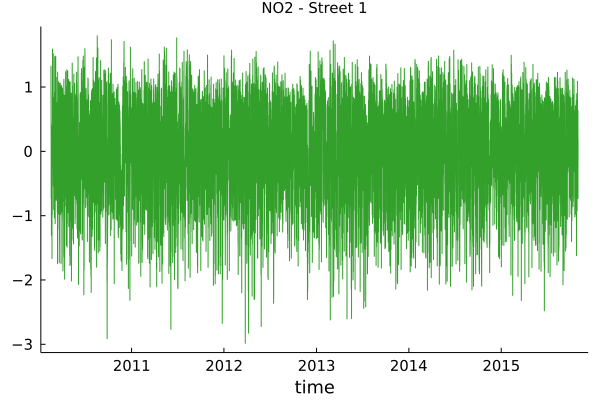}
    \includegraphics[width=0.35\linewidth]{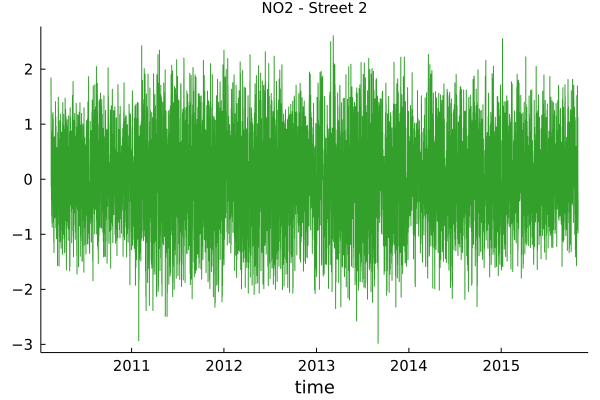} \\
    \includegraphics[width=0.35\linewidth]{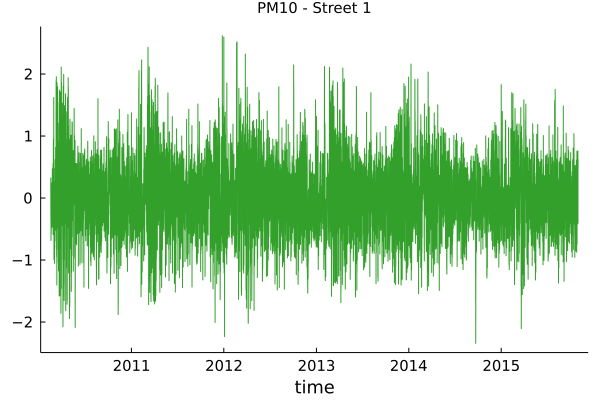}
    \includegraphics[width=0.35\linewidth]{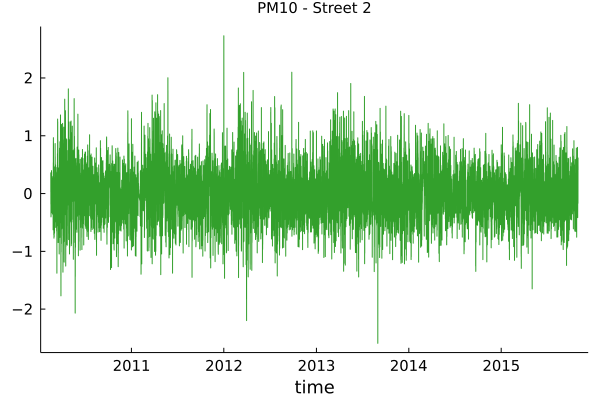}
    \caption{Stockholm air pollution data.}\label{fig:SthlmPollutionData}
\end{figure}

The third dataset contains $\numprint{50,001}$ observations of nitrogen dioxide (NO2) and particulate matter (PM10) pollution at two streets in central Stockholm, giving a four-dimensional time series. The data are measured hourly for the time period February 16, 2010 until October 31, 2015. The missing values are imputed with the \texttt{na.interp} function for seasonal data in R. The series are analyzed in logs through the transformation $\log(x_t-\min \{x_t\}_{t=1}^T + 1 )$ to make the series more normally distributed; we subtract the minimum value and add one to ensure that the series are positive before taking logs (the two PM10 series have a few negative values and the N02 series have some values that are very close to zero). The transformed series are finally filtered through an estimated seasonal $\mathrm{ARIMA}(0,0,0)(1,0,1)_{24}$ to remove the natural 24 hour cyclic trend. The autocorrelation functions of the residuals from this particular seasonal ARIMA model indicate stationarity for all four filtered series. Figure \ref{fig:SthlmPollutionData} plots the transformed and filtered data. The raw data are plotted in Figure \ref{fig:SthlmPollutionRawData} in Appendix \ref{app:additionalresults}.

\subsection{Algorithm, model and prior settings}
The grouped control variates in Section \ref{sec:subsampling_mcmc} are used by dividing the observations into $G=\numprint{1000}$ groups, which gives accurate enough control variates in most of our estimated models. See \cite{salomone2019spectral} for a demonstration on how varying $G$ affects the variance of the estimator. Figure \ref{fig:ExplainGroupAndBlock} illustrates how the $g$th group is chosen by including the $g$th lowest frequency and subsequently systematically sampling every $G$th frequency, ensuring that each group contains periodogram ordinates over the whole frequency spectrum.  We use $10$ randomly selected groups ($1$\% of the data) to estimate the full-data likelihood in each MCMC iteration. For the block pseudo-marginal, we divide the random numbers into $K=10$ blocks (one group per block), resulting in a correlation between the log of the estimators of approximately $0.9$ as discussed in Section \ref{sec:subsampling_mcmc}. For all examples, we sample $\numprint{55000}$ draws from the posterior distribution and discard $\numprint{5000}$ draws as burn-in.

For the ARTFIMA models, we allow for different differencing parameters for each time series. The tempering parameters are restricted to be equal for each time series as the more general model with different $\lambda$ gave estimates of $\lambda$ that are very close across all series in a given dataset.

We use a Minnesota-style prior \citep{doan1984forecasting} for the autoregressive and moving average coefficients, which is a normal prior with diagonal covariance matrix with elements
$$v_{ij,l}=\begin{cases} (\lambda_0/l)^2, &\mbox{if } i = j \\
(\lambda_0 \theta_0\sigma_i/l\sigma_j)^2 & \mbox{if } i \neq j. \end{cases} 
$$
All prior means are set to zero. The hyper-parameter $\lambda_0$ controls how tightly the coefficient of the first lag is concentrated around 0;  we set $\lambda_0=1$. Note that the prior variance decreases with increasing lag length $l$, allowing more shrinkage of coefficients corresponding to lags further back in time. The hyper-parameter $\theta_0$ accounts for the belief that most of the variation in each variable is accounted for by its own lags and is hence set to a value in the range $[0,1)$. We set $\theta_0=0.2$. The ratio $\sigma^2_i/\sigma^2_j$ accounts for the difference in the variability of the variables. Following standard practice, $\sigma^2_i$ is set to the residual variance computed by fitting a univariate AR model to the i$th$ series. Note, however, that we use the above Minnesota prior on the unrestricted parameters in the \citet{ansley1986note} reparametrization so the prior is acting on (a rescaled version of) the partial autocorrelation matrices instead of the original AR and MA coefficients.
We parameterise the covariance matrix $\Sigma_{\varepsilon}$ as a Cholesky factor with a logarithm transform on the diagonals and assign independent $N(0,0.1)$ priors for all elements in this parameterisation. We also use a log-transformation for the single $\lambda$ and assign $N(0,0.1)$. Finally, we assign independent priors $d_k \sim N(0,1)$, for $k=1,\dots, r$. We have verified that our results are robust to the choice of prior.

\subsection{Model comparison and fit}

We perform model selection using the BIC approximation of the log marginal likelihood \citep{kass1995bayes}
\begin{equation*}
    \log p_{\mathrm{BIC}}(\v Y) =  \log p(\v Y|\widehat{\v\theta}) - \frac{k \log n}{2}, 
\end{equation*}
where $k$ is the number of estimated parameters, $n$ is the length of the time series and $\widehat{\v\theta}$ is the maximum likelihood estimate obtained by numerical optimisation. We fit all combinations of models for $p+q \leq 2$ and Table \ref{table:BIC} shows the log marginal likelihood for all models considered for the three datasets introduced in Section \ref{subsec:Data}. The results show that models with tempered fractional differencing give a better model fit for all datasets and AR/MA orders. The two pure MA models, $\mathrm{VARMA}(0,1)$ and $\mathrm{VARMA}(0,2)$, perform very poorly on the temperature data, but improve dramatically when tempered fractional differencing is added. According to the BIC approximation of the marginal likelihood, we conclude the following: 
\begin{itemize}
    \item  $\mathrm{VARTFIMA}(0,2)$ ($14$ parameters) is best for the water velocity dataset.
    \item $\mathrm{VARTFIMA}(2,0)$ ($28$ parameters) is best for the temperature dataset. 
    \item  $\mathrm{VARTFIMA}(1,1)$ and $\mathrm{VARTFIMA}(2,0)$ are best for the pollution dataset. Both models have $47$ parameters.
\end{itemize}

\begin{table}[]
\begin{tabular}{ccrrrcrrrrr}
\hline 
 &  &  & \multicolumn{2}{c}{\vspace{-0.25cm}
} &  & \multicolumn{2}{c}{} &  & \multicolumn{2}{c}{}\tabularnewline
 &  &  & \multicolumn{2}{c}{Water Velocity} &  & \multicolumn{2}{c}{Temperature} &  & \multicolumn{2}{c}{Pollution}\tabularnewline
\cline{4-5} \cline{5-5} \cline{7-8} \cline{8-8} \cline{10-11} \cline{11-11} 
\multicolumn{1}{c}{AR} & \multicolumn{1}{c}{MA} &  & \multicolumn{1}{c}{No TFI} & \multicolumn{1}{c}{TFI} &  & \multicolumn{1}{c}{No TFI} & \multicolumn{1}{c}{TFI} &  & \multicolumn{1}{c}{No TFI} & \multicolumn{1}{c}{TFI}\tabularnewline
1 & 0 &  & $737079$ & $759123$ &  & $327097$ & $334122$ &  & $363760$ & $366022$\tabularnewline
0 & 1 &  & $588297$ & $759457$ &  & $61320$ & $332888$ &  & $306068$ & $365658$\tabularnewline
2 & 0 &  & $749650$ & $761200$ &  & $335201$ & $\textbf{335757}$ &  & $365522$ & $\textbf{366266}$\tabularnewline
0 & 2 &  & $621765$ & $\textbf{761786}$ &  & $93256$ & $333948$ &  & $325717$ & $366142$\tabularnewline
1 & 1 &  & $758838$ & $761305$ &  & $333582$ & $335647$ &  & $365762$ & $\textbf{366267}$\tabularnewline
 &  &  & \vspace{-0.35cm}
 &  &  &  &  &  &  & \tabularnewline
\hline 
\end{tabular} 
\vspace{0.1cm}
\caption{BIC approximation of the log marginal likelihood for different models for each of the three datasets in Section \ref{subsec:Data}. A higher value indicates a better model fit.  The AR and MA columns indicate the lag order in the AR and MA component. The No TFI and TFI columns indicate if the process has tempered fractional differencing. The model with the largest marginal likelihood for each dataset is marked in bold font. Both the VARTFIMA(2,0) and VARTFIMA(1,1) are in bold font for the Pollution data since the difference between them is `not worth more than a bare mention' on the modified Jeffreys' scale of evidence in \cite{kass1995bayes}.}
\label{table:BIC}
\end{table}

We implemented spectral subsampling MCMC successfully in all cases except for 
\begin{itemize}
    \item Water velocity: $\mathrm{VARMA}$(0,2).
    \item Temperature: $\mathrm{VARMA}$(0,1), $\mathrm{VARMA}$(0,2) and $\mathrm{VARTFIMA}$(0,2).
    \item Pollution: $\mathrm{VARMA}$(0,2), $\mathrm{VARMA}$(1,1), $\mathrm{VARTFIMA}$(2,0), $\mathrm{VARTFIMA}$(1,1).
\end{itemize}

The reason for the occasional failure of spectral subsampling MCMC for some models is that the control variates do not reduce the variance of the likelihood estimator sufficiently, even after grouping. The situation improves for longer time series, since the control variates typically improve with more data \citep{quiroz2019speeding}. To illustrate this, we consider the $\mathrm{VARTFIMA}$(0,2) model for the Swedish temperature data, where spectral subsampling was successful. The top left graph in Figure \ref{fig:SwedTempVARTFIMA20fail} shows the standard deviation of the log-likelihood estimator over the MCMC iterations, which is the key quantity determining the efficiency of pseudo-marginal algorithms \citep{pitt2012some}; the optimal standard deviation from Lemma S8 in \citet{quiroz2020block} is also indicated in the graph. The variability of the estimator does not exceed the optimal value and spectral subsampling works well, as exemplified by the good mixing of the MCMC chain in the bottom left graph of Figure \ref{fig:SwedTempVARTFIMA20fail} for the AR parameter of the first series on its own first lag, $\Phi_{11}^{(1)}$. The top right graph in Figure \ref{fig:SwedTempVARTFIMA20fail} shows the same quantity, but from spectral subsampling on a posterior based on only the last $60001$ data points. The variability of the estimator is initially below the optimal value and the sampler mixes well, but the algorithm eventually moves to a region of the parameter space where the estimator is often much too variable, leading to low acceptance rates and the MCMC chain getting stuck for long spells (bottom right figure).

\begin{figure}
    \includegraphics[width=0.45\linewidth]{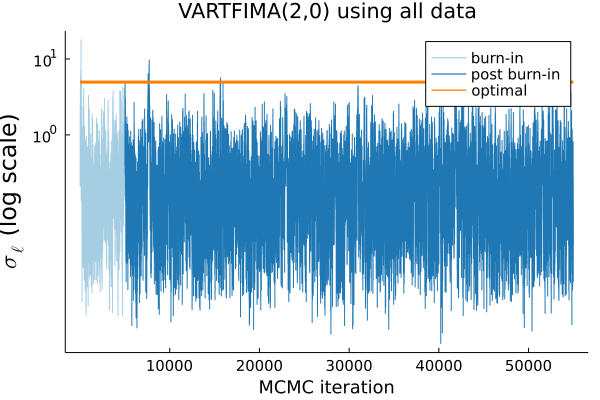}\includegraphics[width=0.45\linewidth]{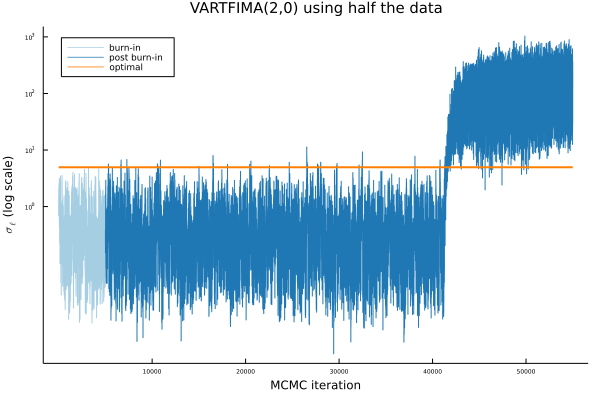}
    
    \includegraphics[width=0.45\linewidth]{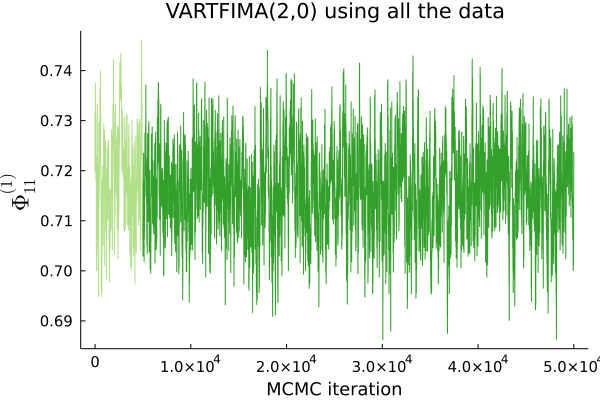}\includegraphics[width=0.45\linewidth]{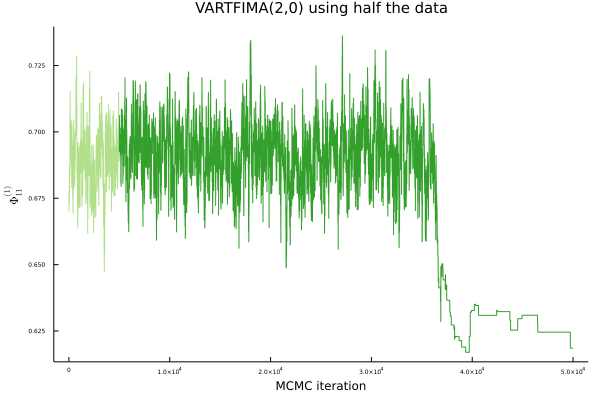}
    \caption{Standard deviation of the log-likelihood estimator over the MCMC iterations for the VARTFIMA(2,0) model fitted to the Swedish termperature dataset using all the data (top left) versus when only the last $60001$ observations was used in the estimation (top right). The optimal standard deviation \citep[Lemma S8]{quiroz2020block} is marked out with an orange line. The graphs in the bottom row of the figure show the MCMC chains for the AR parameter of the first series on its own first lag, $\Phi_{11}^{(1)}$, for the respective model.}\label{fig:SwedTempVARTFIMA20fail}
\end{figure}

Since the control variates are based on a Taylor expansion around the posterior mode, they will work poorly when the posterior has multiple modes or a long ridge of fairly constant density. Such shapes are typical in models with local non-identification. VARMA models are, for example, known to have local non-identification problems resulting from near cancellation of roots in the AR and MA polynomials \citep{chib1994bayes}, i.e.\ the likelihood function is flat in the direction of certain linear combination of the parameters. There are also potential identification issues in the MA part of VARMA models \citep[Ch. 7]{lutkepohl2013introduction}. Adding tempered fractional differencing provides more flexibility, but also additional local identification problems since the fractional differencing parameter $d$ becomes non-identified when $\lambda \rightarrow \infty$. Prior distributions alleviate these identification issues to some extent, but the posteriors in the models considered here are nevertheless very challenging. Luckily, subsampling MCMC fails on the worst fitting models in the Water velocity and Swedish temperature data, but unfortunately also gets stuck on the two best fitting models for the Pollution data. For the Pollution dataset we will therefore only present results using MCMC on the full data without subsampling, but using both the exact time domain likelihood as well as the approximate Whittle likelihood.

Figure \ref{fig:WaterVelocitySpectral} assesses the fit of the VARTFIMA(0,2) for the Water velocity data by how well the predictive distribution captures the univariate periodogram data. The predictive distribution is computed by, for each posterior draw $\v \theta$, simulating $100$ periodogram observations using the Whittle approximation
\begin{equation*}
I_{j,T}(\omega) \vert \v \theta \overset{\mathrm{indep}}{\sim} \mathrm{Expon}(f_{jj,\boldsymbol{\theta}}(\omega)),
\end{equation*}
where $I_{j,T}$ is the periodogram data for the $j$th time series. A nonparametric multitaper estimate \citep{barbour2014psd} is also shown as a grey line in the figures. The predictive distribution in Figure \ref{fig:WaterVelocitySpectral} captures the periodogram data quite well. The multitaper estimate suggest that there are some peaks in the spectral density not captured by the VARTFIMA(0,2) model. A VARTFIMA with higher AR and MA orders would be able to also fit those spectral bumps, but the peaks occur at different periods in the two series and may just be an artefact of the large number of observations in the dataset. Figures \ref{fig:SwedTempSpectral} and \ref{fig:SthlmPollutionSpectral} in Appendix \ref{app:additionalresults} present the same results for the two other datasets. The figures show that some seasonality seems to have survived the pre-processing of the datasets. An alternative is fitting a multivariate seasonal ARTFIMA model to the original data, but that is not pursued here.

\begin{figure}
    \includegraphics[width=0.45\linewidth]{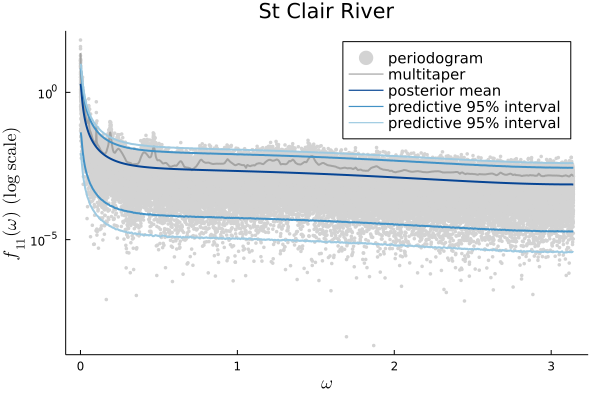}
    \includegraphics[width=0.45\linewidth]{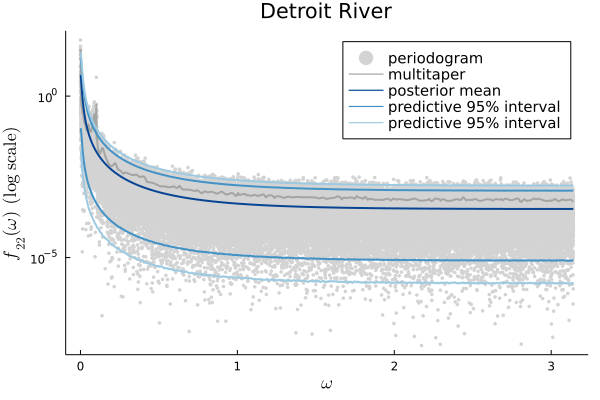}
    \caption{Posterior predictive fit of the univariate periodogram data for the VARTFIMA(1,1) model fitted to the Water velocity data. The predictive intervals are obtained by simulation from the asymptotic Whittle distribution $\mathcal{I}(\omega)\sim\mathrm{Expon}(f_{\v\theta}(\omega))$ with parameters $\v\theta$ drawn from the VARTFIMA(1,1) posterior.}\label{fig:WaterVelocitySpectral}
\end{figure}

\subsection{Efficiency of spectral subsampling MCMC}

To measure the computational advantage of subsampling we use a performance measure that takes into account both the cost of estimating the likelihood and the inefficiency of the MCMC chain. The computational time (CT) of an algorithm is defined as
\begin{equation}
    \mathrm{CT} \equiv \mathrm{IACT} \times \text{Computing time for a single iteration},
\end{equation}
where $\mathrm{IACT} \equiv 1 +2\sum_{k=1}^\infty \rho_k$ is the integrated autocorrelation time of the MCMC chain, and $\rho_k$ is the autocorrelation at lag $k$ of the posterior draws. CT measures the execution time for obtaining the equivalent of a single iid draw from the posterior. To obtain an implementation independent measure, the computing time is set proportional to the number of density evaluations in a run of the algorithm, including the evaluations needed to construct the control variate for subsampling. We use the \texttt{coda} package \citep{plummer2006coda} in \texttt{R} to estimate the integrated auto-correlation time; see \cite{quiroz2019speeding} for more details.

Table \ref{tab:Relative_CT} shows the relative computational time (RCT) of MCMC on the full datasets in relation to spectral subsampling MCMC. The RCT is defined as the ratio between the CT of full-data MCMC and that of spectral subsampling MCMC. Hence, values larger than one mean that spectral subsampling is more efficient when taking into account both the computing cost and the sampling efficiency. The results show that our subsampling algorithm is between 68-125 times faster than MCMC on the full dataset when RCT is the measure of computational efficiency. Similar speed-ups are observed on the Pollution dataset for the models where subsampling MCMC did not get stuck.
\begin{table}[]
\begin{tabular}{lcrrrcrrr}
\hline 
 &  &  & \multicolumn{2}{c}{\vspace{-0.25cm}
} &  & \multicolumn{2}{c}{} & \tabularnewline
\multicolumn{1}{l}{Dataset} & \multicolumn{1}{c}{} &  & \multicolumn{1}{c}{Model} & \multicolumn{1}{c}{} & Min & \multicolumn{1}{c}{Mean} & \multicolumn{1}{c}{Max} & \tabularnewline
\cline{1-1} \cline{4-4} \cline{6-8} \cline{7-8} \cline{8-8} 
Water velocity &  &  & VARTFIMA(0,2) &  & 87 & 98 & 125 & \tabularnewline
Sweden temperature &  &  & VARTFIMA(2,0) &  & 68 & 89 & 114 & \tabularnewline
\hline 
\end{tabular}
\caption{Relative computational time (CT) of comparing MCMC using the full dataset to spectral subsampling MCMC. The value $1$ indicates that spectral subsampling MCMC and MCMC are equally efficient, and values larger than 1 indicate that spectral subsampling MCMC is the better algorithm. The results for the pollution dataset are not shown as subsampling MCMC got stuck for the best models for that dataset (VARTFIMA$(2,0)$ and VARTFIMA$(1,1)$, see Table \ref{table:BIC}).}\label{tab:Relative_CT}
\end{table}

\subsection{Accuracy of spectral subsampling for the Whittle posterior}
This subsection explores how well spectral subsampling MCMC approximates the posterior based on the Whittle likelihood for the full data. The next section investigates how well the Whittle likelihood for the full data approximates the time domain likelihood in finite samples.

\begin{figure}
    \includegraphics[width=0.9\linewidth]{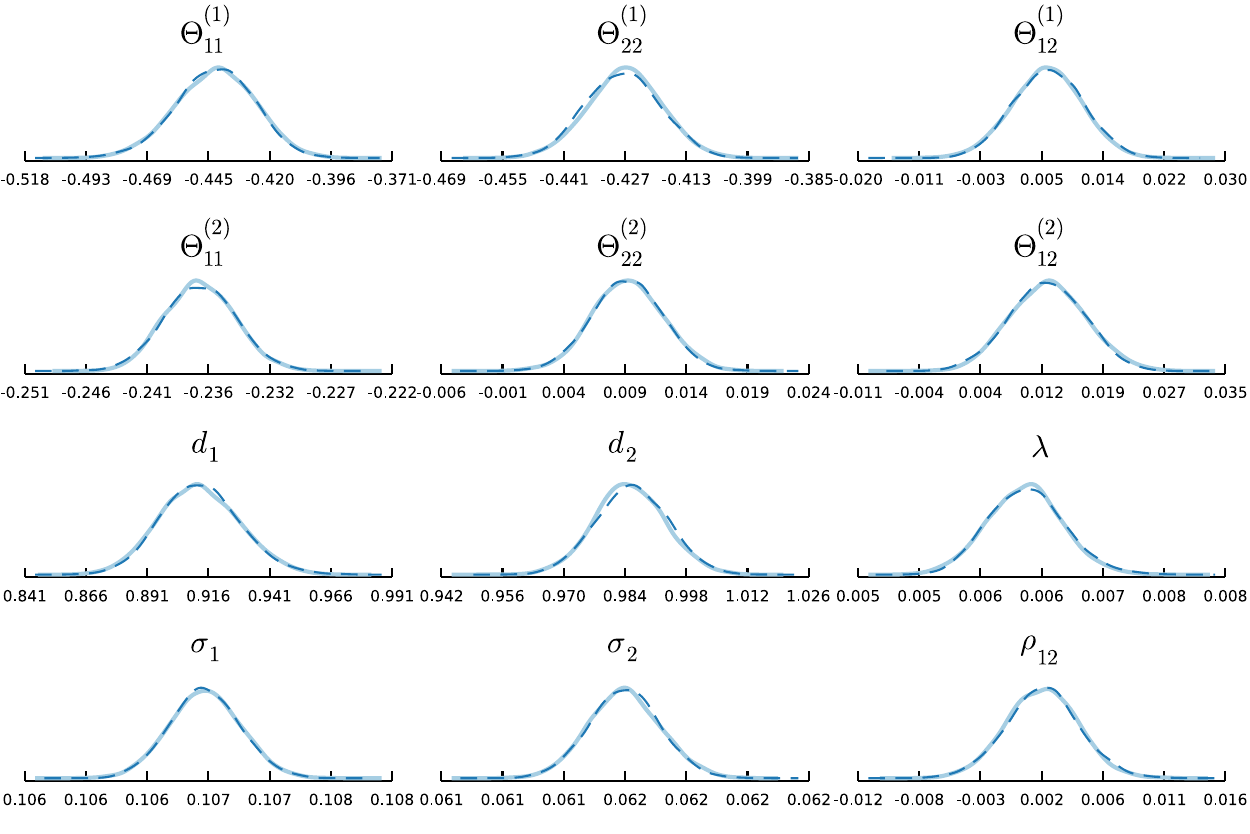}
    \caption{Kernel density estimates of a subset of the marginal posterior densities for the VARTFIMA(0,2) model fitted to the Water velocity data. The solid densities are from MCMC using the Whittle posterior on the whole dataset, and the dashed densities are obtained from spectral subsampling MCMC.}\label{fig:WaterVelocityKDEVARTFIMA02}
\end{figure}

Figure \ref{fig:WaterVelocityKDEVARTFIMA02} and \ref{fig:SwedTempKDEartfima20} compare kernel density estimates of the posterior distribution of a subset of the parameters using subsampling MCMC and full-data MCMC. We conclude that the subsampling MCMC algorithm provides very similar answers, and Table \ref{tab:Relative_CT} shows it is up to two orders of magnitude faster.

\begin{figure}
    \includegraphics[width=0.75\linewidth]{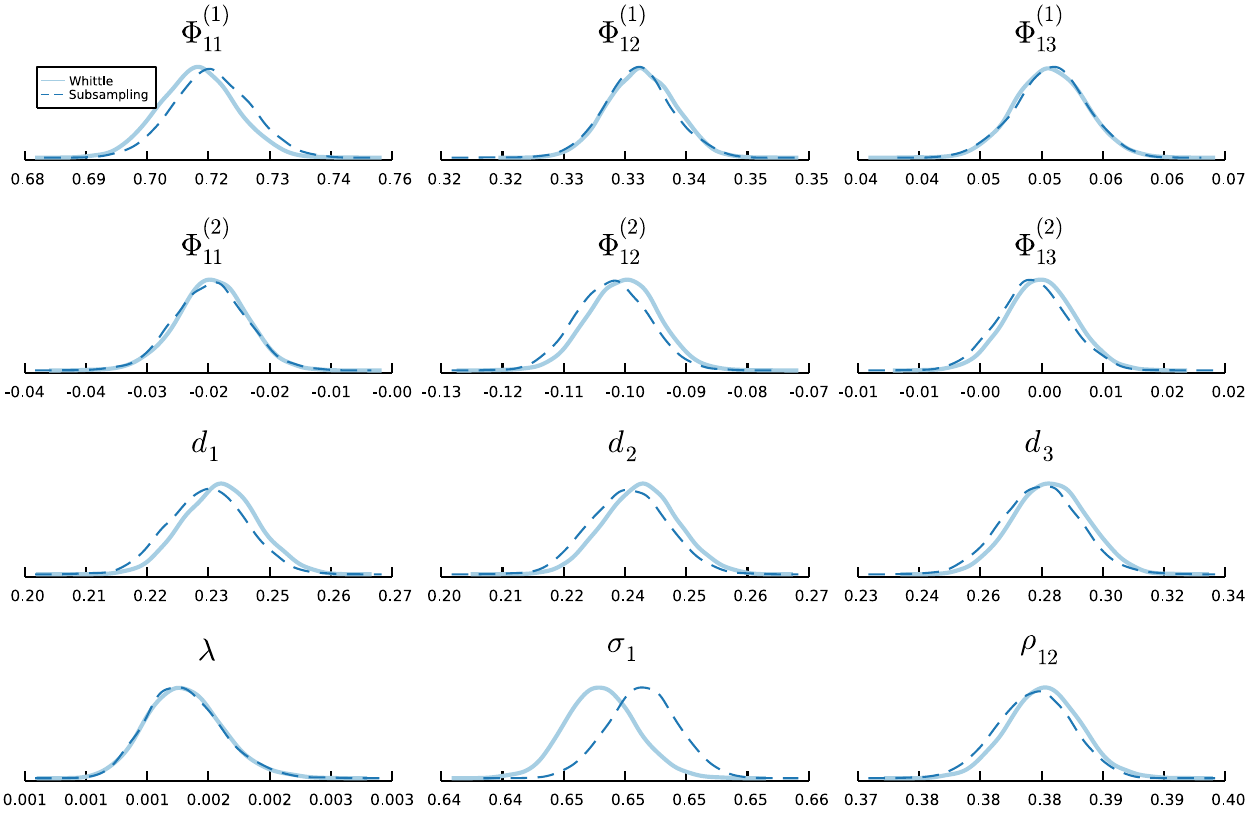}
    \caption{Kernel density estimates of a subset of the marginal posterior densities for the VARTFIMA(2,0) model fitted to the Swedish temperature data. The solid densities are from MCMC using the Whittle posterior on the whole dataset, and the dashed densities are obtained from spectral subsampling MCMC.}\label{fig:SwedTempKDEartfima20}
\end{figure}


The individual parameters in time series models are seldom of practical interest. We will therefore now explore how well our method approximates the posterior of summary quantities based on the spectral density matrix that are often used in practical work. The \emph{coherence} between two time series $X_{it}$  and $X_{jt}$ at frequency $\omega$ is defined as
\begin{equation*}
    \mathcal{K}_{ij}(\omega)=
    \frac{f_{ij}(\omega)}
    {\sqrt{f_{ii}(\omega)f_{jj}(\omega)}}.
\end{equation*}
The squared coherence 
\begin{equation*}
    0\leq \vert \mathcal{K}_{ij} \vert ^2 \leq 1,
\end{equation*}
measures the linear correlation between the pair of time series $X_{it}$ and $X_{jt}$ at frequency $\omega$. Expressing the complex-valued cross spectral density $f_{ij}(\omega)$ in polar form as $$f_{ij}(\omega) =  \vert f_{ij}(\omega) \vert \exp(\ci\varphi(\omega)),$$ where the phase spectrum $$\varphi_{ij}(\omega) = \mathrm{arg}(f_{ij}(\omega)),$$ measures the time shift of the signal at frequency $\omega$, the \emph{time delay} from variable $X_{it}$ to variable $X_{jt}$ is measured by $-\varphi_{ij}(\omega)/\omega$ \citep{wei1990time}. See \citet{geweke1982measurement} and \citet{ashby2019statistical} for connections between these frequency based measures of linear association and direction, and cross-correlation functions and Granger causality measures in the time domain.

\begin{figure}
    \includegraphics[width=0.8\linewidth]{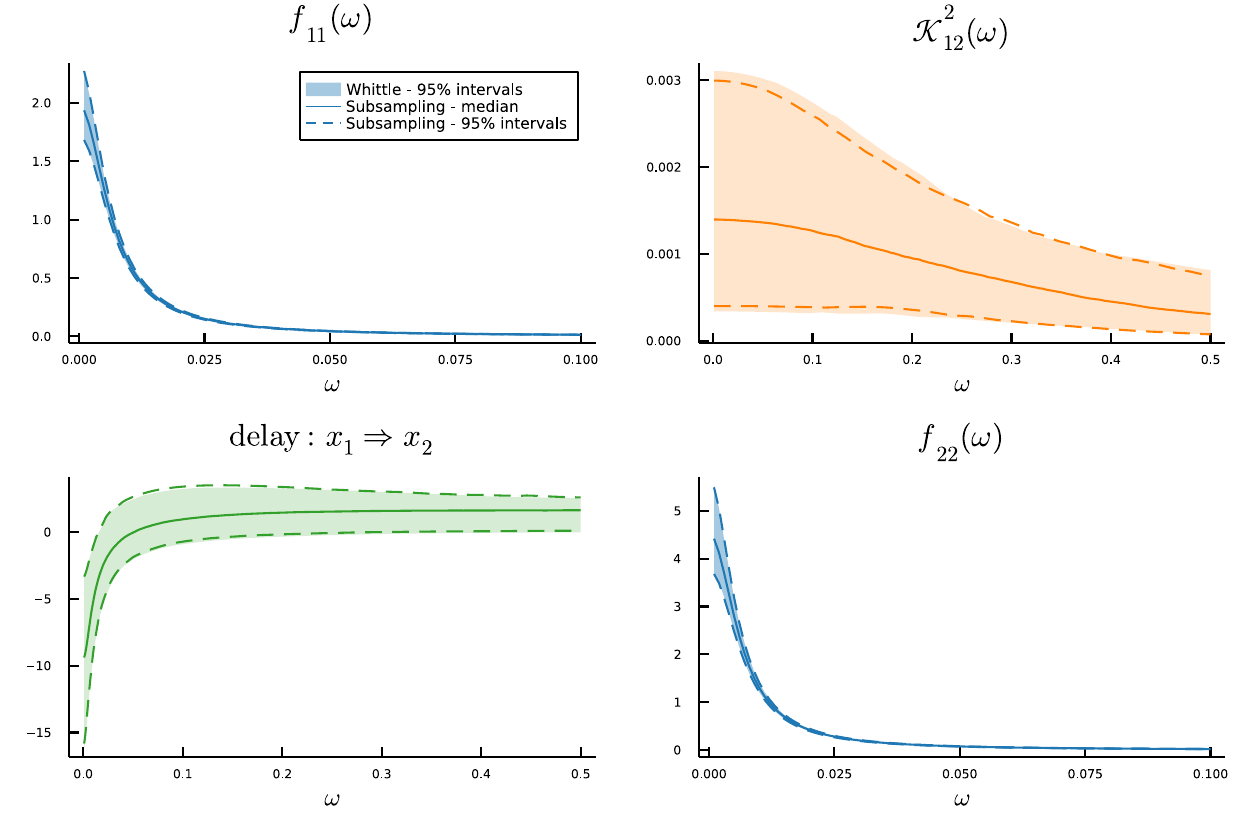}
    \caption{Posterior for the spectral density matrix in the VARTFIMA(0,2) model fitted to the Water velocity data. The plots on the diagonal are the marginal spectral densities for each time series. The plots above the diagonal are the squared coherence, and the plots below the diagonal are the time delays from the phase spectrum. The dashed lines in each subplot displays the posterior median and 95\% credible intervals from spectral subsampling MCMC. The shaded regions are the 95\% credible intervals from the Whittle posterior on the whole dataset.}\label{fig:WaterVelocityCoherence}
\end{figure}

Figure \ref{fig:WaterVelocityCoherence} plots the posterior distribution of these spectral density matrix quantities for the Water velocity dataset. The figure shows the marginal spectral densities (graphs on the diagonal) of each time series, squared coherence (above diagonal) and phase/time delay (below diagonal). Each figure displays the posterior mean and $95\%$ credible intervals from spectral subsampling MCMC as lines. The $95\%$ credible intervals from MCMC on the Whittle posterior based on the full dataset is shown as shaded regions. As expected from the previously shown parameter posteriors, we again see that subsampling gives virtually no distortion with respect to the Whittle posterior on the full dataset. The same is true for the temperature data; see Figure \ref{fig:SwedTempCoherence} in Appendix \ref{app:additionalresults}.

The interpretation of the lag delay is clearest in systems without feedback loops. This is likely to be the case only in the Water velocity data where the Denver River location ($x_2$) is located downstreams of the St Clair River location ($x_1$). The top right graph in Figure \ref{fig:WaterVelocityCoherence} shows, however, that there is essentially no coherence between the two locations at any frequency, so the lag delay is insignificant. This is most probably because the two locations are quite far apart (approx $100$ kilometers) and separated by Lake St Clair. The VARMA(1,1) model fitted in the next subsection picks up a sizeable coherence only at the very lowest frequencies and a positive delay at those frequencies (see Figure \ref{fig:WaterVelocityCoherenceVarma}), which is the expected sign since Denver River is downstream of St Clair River.


\subsection{Accuracy of the Whittle posterior for the time domain posterior}
The previous subsection shows that spectral subsampling MCMC gives very accurate approximations to the Whittle posteriors based on the full datasets. The Whittle likelihood can, however, be a poor approximation to the exact time domain likelihood \citep{contreras2006note}, at least for shorter time series. However, the original motivation for spectral subsampling MCMC is settings where the time series are very long, and this section demonstrates that the Whittle approximation is excellent in the three datasets in Section \ref{subsec:Data}.

\begin{figure}
    \includegraphics[width=0.8\linewidth]{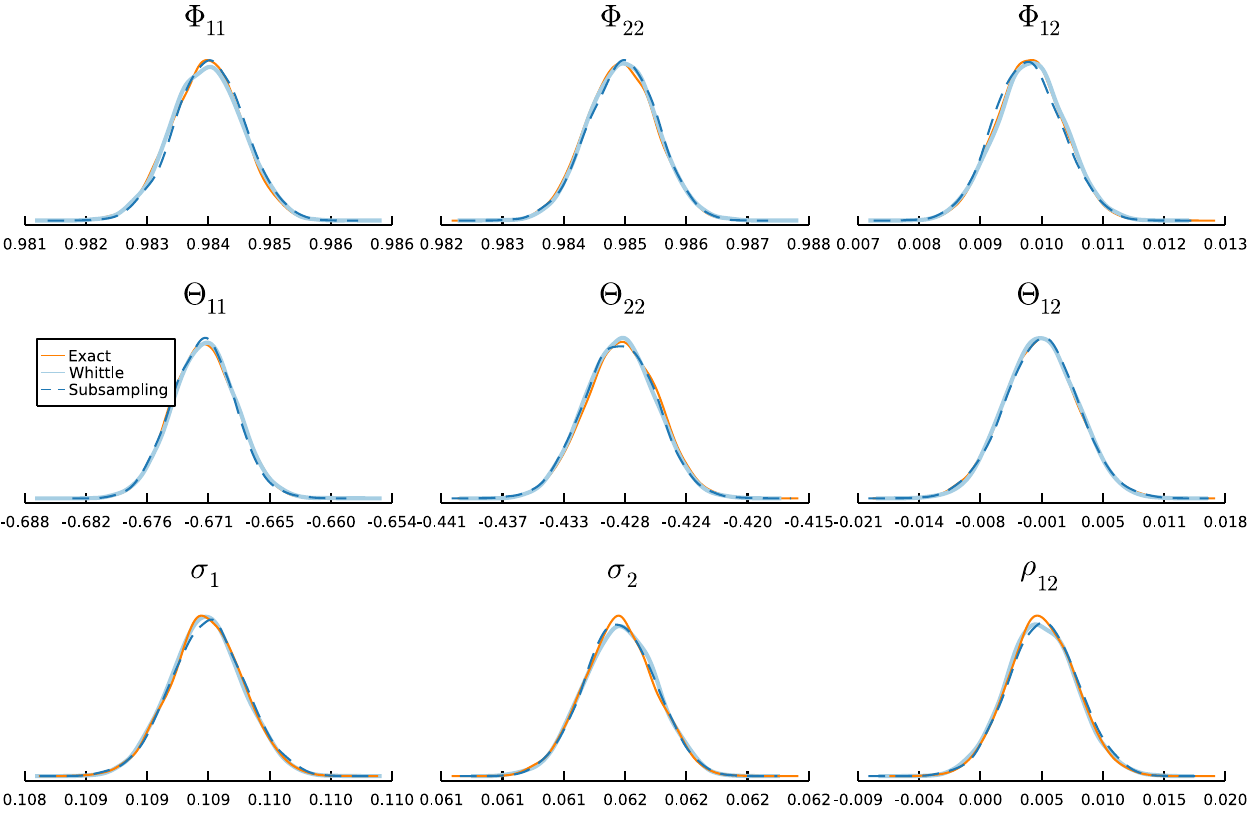}
    \caption{Kernel density estimates of a subset of the marginal posterior densities for the VARMA(1,1) model fitted to the Water velocity data. The solid orange densities are from MCMC on the exact time domain posterior, the solid blue densities are from MCMC using the Whittle posterior on the whole dataset, and the dashed densities are obtained from spectral subsampling MCMC.}\label{fig:WaterVelocityKDEarma11}
\end{figure}

\begin{figure}
    \includegraphics[width=0.7\linewidth]{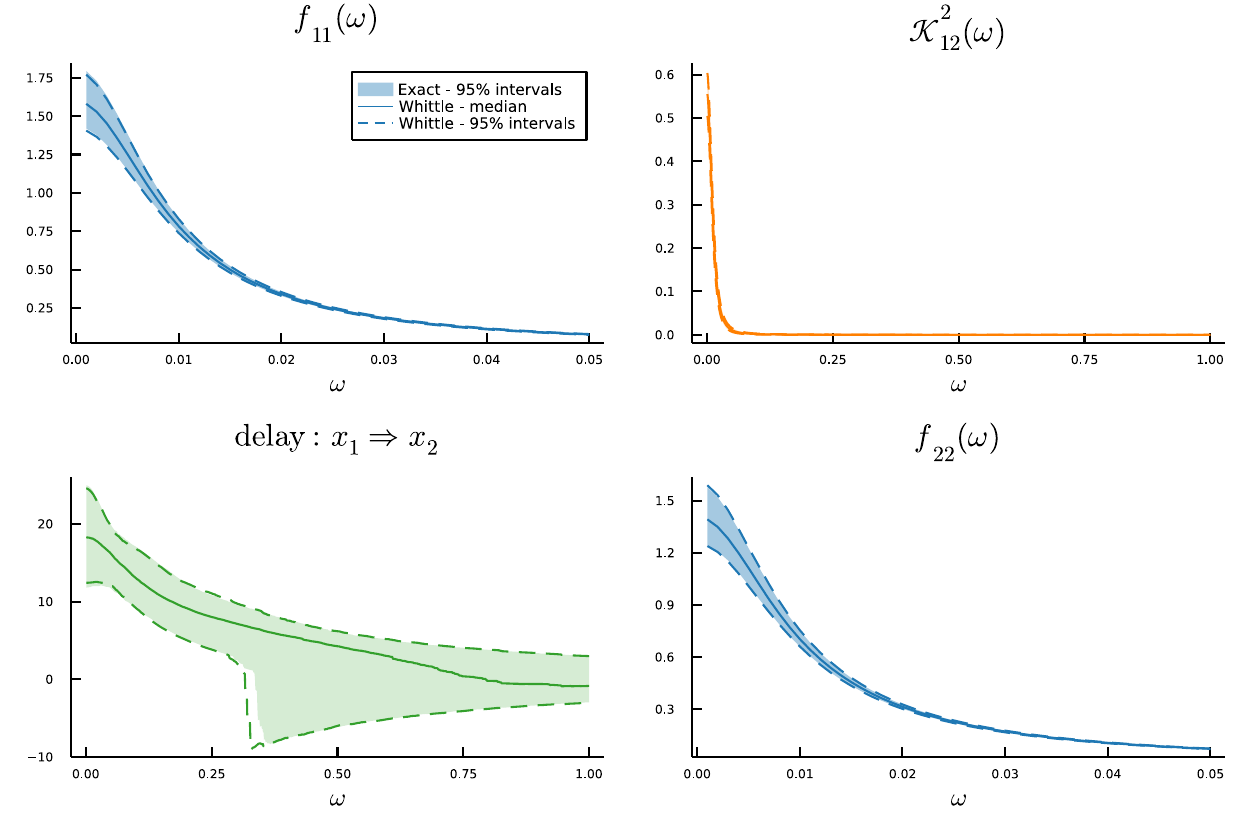}
    \caption{Posterior for the spectral density matrix in the VARMA(1,1) model fitted to the Water velocity data. The plots on the diagonal are the marginal spectral densities for each time series. The plot above the diagonal is the squared coherence, and the plot below the diagonal is the phase spectrums. The dashed lines in each subplot displays the posterior median and 95\% credible intervals from the Whittle posterior on the whole dataset. The shaded regions are the 95\% credible intervals from MCMC on the exact time domain posterior on the whole dataset.}\label{fig:WaterVelocityCoherenceVarma}
\end{figure}

\begin{figure}
    \includegraphics[width=0.8\linewidth]{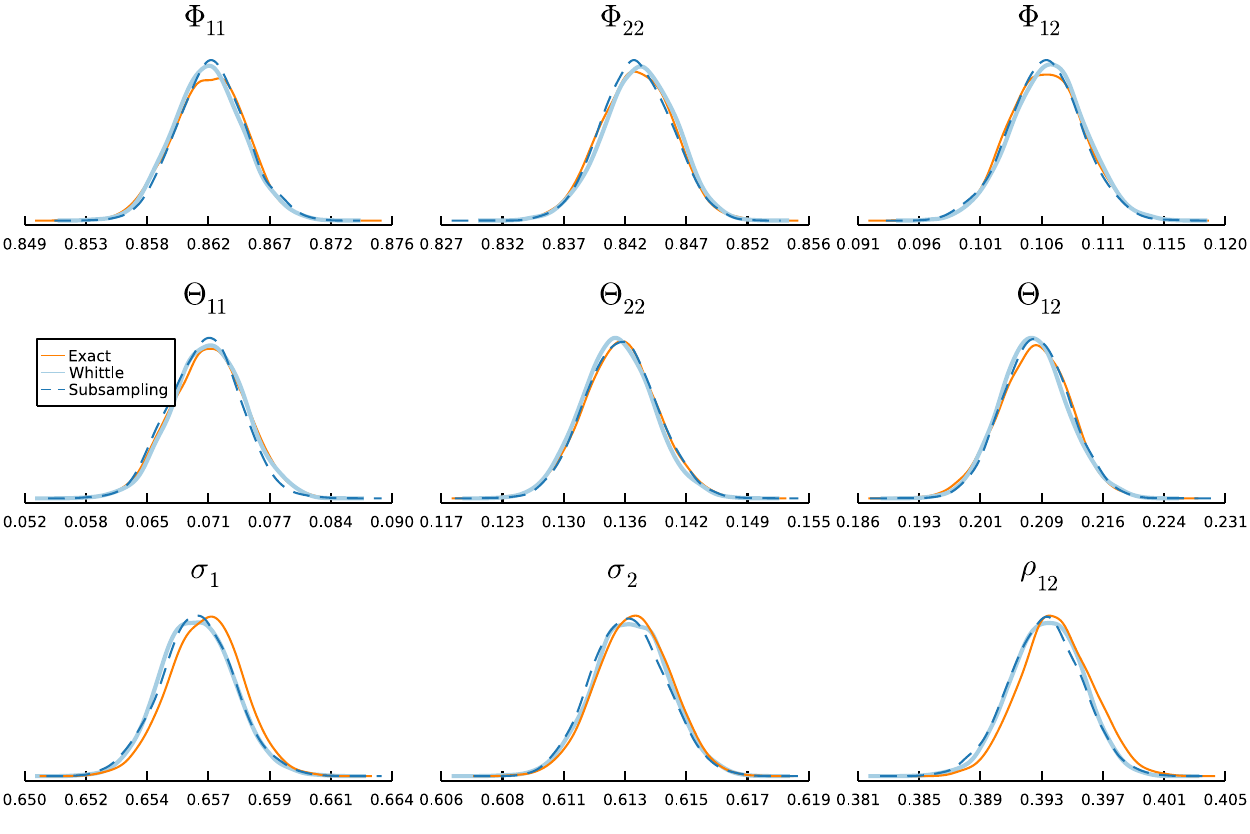}
    \caption{Kernel density estimates of a subset of the marginal posterior densities for the VARMA(1,1) model fitted to the Swedish temperature data. The solid orange densities are from MCMC on the exact time domain posterior, the solid blue densities are from MCMC using the Whittle posterior on the whole dataset, and the dashed densities are obtained from spectral subsampling MCMC.}\label{fig:SwedTempKDEarma11}
\end{figure}

\begin{figure}
    \includegraphics[width=0.8\linewidth]{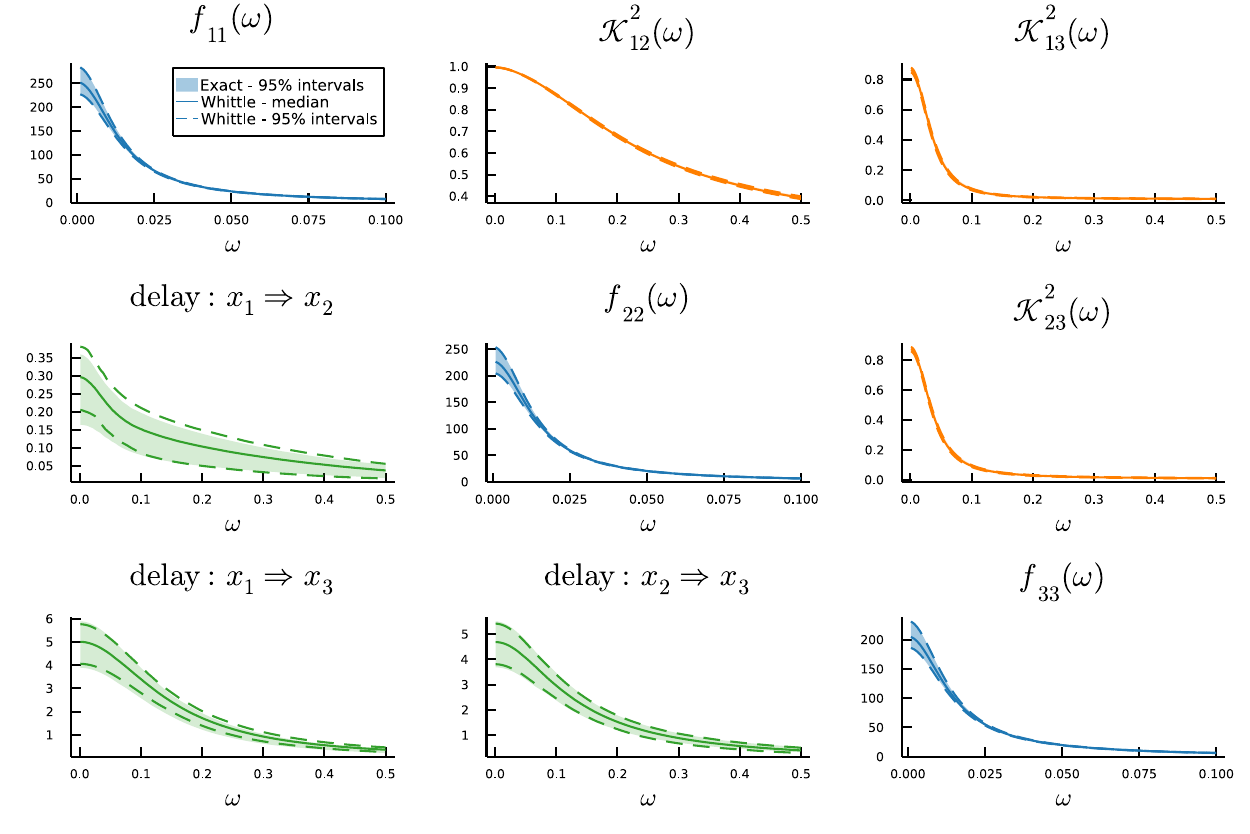}
    \caption{Posterior for the spectral density matrix in the VARMA(1,1) model fitted to the Swedish temperature data. The plots on the diagonal are the marginal spectral densities for each time series. The plot above the diagonal is the squared coherence, and the plot below the diagonal is the phase spectrums. The dashed lines in each subplot displays the posterior median and 95\% credible intervals from the Whittle posterior on the whole dataset. The shaded regions are the 95\% credible intervals from MCMC on the exact time domain posterior on the whole dataset.}\label{fig:SwedTempCoherenceVarma}
\end{figure}

The time domain likelihood is only computationally feasible for VARMA models, and hence we illustrate these results using a $\mathrm{VARMA}(1, 1)$ model for each dataset. Figure \ref{fig:WaterVelocityKDEarma11} shows the kernel density estimates of the posterior distribution of a subset of the parameters. Figure \ref{fig:WaterVelocityCoherenceVarma} plots the time domain and Whittle posteriors of the spectral density matrix based on the full Water velocity dataset. The results from the Whittle posterior are nearly indistinguishable from the exact time domain posteriors. Also, there are virtually no differences between the time domain posteriors and the whittle posteriors in the two other datasets; see Figure \ref{fig:SwedTempKDEarma11} and \ref{fig:SwedTempCoherenceVarma} for the temperature data and Figure \ref{fig:SthlmPollutionKDEarma11} in Appendix \ref{app:additionalresults} for the pollution data. In summary, we conclude from this and the previous subsection that spectral subsampling MCMC gives an accurate approximation of the posterior based on the exact time domain likelihood. 

We have implemented the exact time domain likelihood for VARMA models using the state space representation in the Python package \texttt{statsmodels} and fair timing comparisons with spectral subsampling MCMC are therefore difficult, but it is well known that time domain likelihoods are substantially more costly than their frequency domain (Whittle) counterparts, even without the extra speed-up from subsampling.

\section{Comparing ARTFIMA and VARTFIMA}\label{sec:univariate}

Our paper focuses on the computational performance of spectral subsampling MCMC on challenging multivariate processes, and we leave a detailed analysis of the empirical performance of the VARTFIMA model to a separate paper. However, a reviewer suggested comparing the proposed multivariate VARTFIMA model with univariate ARTFIMA models \citep{Sabzikar2019} for each series separately, and we present some results on this aspect below using the Swedish temperature data. 

\begin{figure}
    \includegraphics[width=0.7\linewidth]{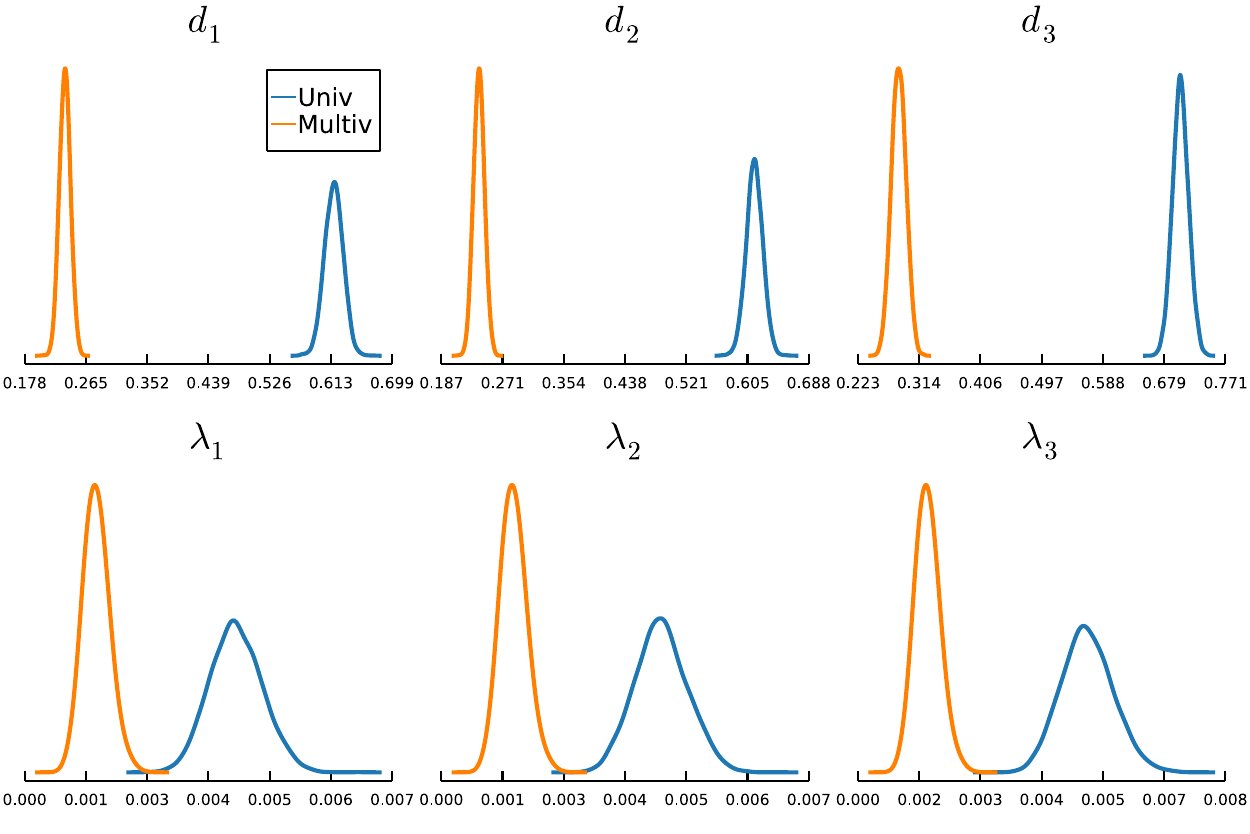}
    \caption{Comparing the marginal posteriors for the fractional differencing and tempering parameters in the multivariate VARTFIMA(2,0) model (blue) to univariate ARTFIMA models (orange) fitted to the Swedish temperature data. The multivariate model is restricted to have the same $\lambda$ for each time series and the posterior density for this common $\lambda$ is therefore repeated in each graph on the second row in the figure. Spectral subsampling MCMC is used for both models.}\label{fig:SwedTempKDEUniv}
\end{figure}

First, we compute the BIC approximation of the log marginal likelihood for a univariate version of the preferred VARTFIMA(2,0) model by restricting the diagonals of all $\Phi$ and $\Theta$ and $\Sigma$ to zero in the estimation, but allowing for different tempering parameter and fractional differencing for each series. This corresponds to fitting ARTFIMA models to each series separately, but still obtaining a single BIC that approximates the marginal likelihood of all time series; these BIC values are therefore comparable to the ones from the multivariate VARTFIMA models. The BIC for this univariate version is $312824$ which should be compared to the BIC of $335757$ for the multivariate VARTFIMA(2,0) in Table \ref{table:BIC}. The massive improvement in BIC from jointly modeling the time series in the VARTFIMA(2,0) model is expected since the series are cross-correlated; see Figure \ref{fig:SwedTempCoherence}.

Furthermore, the fractional differencing parameters can be quite different when modeling the time series jointly as in VARTFIMA compared to univariate ARTFIMA models. Figure \ref{fig:SwedTempKDEUniv} compares the marginal posteriors of the fractional differencing and tempering parameters in the multivariate VARTFIMA(2,0) model to univariate ARTFIMA(2,0) models in the Swedish temperature data. The VARTFIMA model in Figure \ref{fig:SwedTempKDEUniv} is restricted to have the same $\lambda$ in the three series; the results are quite similar when separate $\lambda$ are used. While the tempering parameters are rather similar for the univariate and multivariate models, the fractional differencing is substantially reduced for all three series when the multivariate model is used.

Finally, we compare the forecasting performance of VARTFIMA(2,0) model to the univariate ARTFIMA(2,0). The models are estimated on the same data as used above, i.e. the period from February 1, 2008 at 00:00 hours to March 25, 2022 at 16:00 hours, and the forecasts are evaluated over subsequent period from March 25, 2022 at 17:00 hours to May 1, 2022 at 6:00 hours, making up a test set with $878$ number of observations. Since the training data is very large and the estimates are precise, we do not update the estimates as we move across the test data. 

There are several ways of computing the forecasts from the VARTFIMA model. One way is by computing the autocovariance matrix function from an inverse FFT of the spectral density matrix at the estimated VARTFIMA parameter and then using the conditioning properties of the multivariate normal distribution to get the forecast, see \citet{Hamilton1995} for the theory and \citet{mcleod2008algorithms} for an R package implementation. This is an elegant and general approach, but does not scale well to large data sets: the \texttt{artfima} package in R, which uses the \citet{mcleod2008algorithms} R package for forecasting, crashes when using more than 20\% of the Swedish temperature data on a 32 GB RAM Linux machine. 

We take a more direct and computationally faster approach here by approximating the $\mathrm{VARTFIMA}(p,q)$ with a $\mathrm{VARMA}(p_\star,q)$ for a sufficiently large $p_\star$. The $\mathrm{VARTFIMA}(p,q)$ model is
\begin{equation}
    \big(I_r - \Phi_1L - \ldots - \Phi_p L^p \big)\Delta^{\mathbf{d},\v \lambda}(\mathbf{Y}_t-\boldsymbol \mu) = \Theta(L)\boldsymbol{\varepsilon}_t,
\end{equation}
where from the definition of $\Delta^{\mathbf{d},\v \lambda}$ we can express
\begin{equation}
\Delta^{\mathbf{d},\v \lambda} = I_r + B_1L + B_2L^2 + \ldots
\end{equation}
with $B_j = \mathrm{Diag}(b_j^{d_1,\lambda_1},\ldots,b_j^{d_r,\lambda_r})$ and $b_j^{d_k,\lambda_k} \equiv (-1)^j \binom{d_k}{j}e^{-\lambda_k j}$ for $j=1,2,\ldots$. We can therefore write
\begin{equation}
    \big(I_r - \Phi_1L - \ldots - \Phi_p L^p \big)\big(I_r + B_1L + B_2L^2 + \ldots\big) = I_r - \Pi_1L - \Pi_2L^2 - \ldots\big),
\end{equation}
where the $\Pi_k$ can be found by equating the coefficients term by term for each $L^k$ to obtain
\begin{equation}
    \Pi_k = \sum_{j=1}^p \Phi_j B_{k-j} - B_k,
\end{equation}
by defining $B_0 =  I_r$ and $B_j = \v 0$ for $j<0$. 
Since $b_j^{d_k,\lambda_k} \rightarrow 0$ as $j\rightarrow \infty$ exponentially fast eventually, if $\lambda_k>0$ \citep{Sabzikar2019} we have that the elements of $\Pi_k$ goes to zero as $k$ increases. We can therefore truncate this infinite order lag matrix polynomial at some lag $p_\star$ to get a $\mathrm{VARMA}(p_\star,q)$ approximation
\begin{equation}
   \Pi(L)(\mathbf{Y}_t-\boldsymbol \mu) = \Theta(L)\boldsymbol{\varepsilon}_t,
\end{equation}
where $I_r - \Pi_1 L - \ldots - \Pi_{p_{\star}}L^{p_{\star}}$. Forecasts $k$ steps ahead can now be obtained with established techniques for VARMA models \citep{tsay2013multivariate}; we use the \texttt{MTS} package in R. Figure \ref{fig:SwedTempForecastApprox} shows that $p_{\star}=10$ is more than sufficient for approximating the $\mathrm{VARTFIMA}(2,0)$ model used in the forecast evaluation below on the Swedish temperature data in that the forecast do not change beyond $p_{\star}=10$. 

\begin{figure}
    \includegraphics[width=0.9\linewidth]{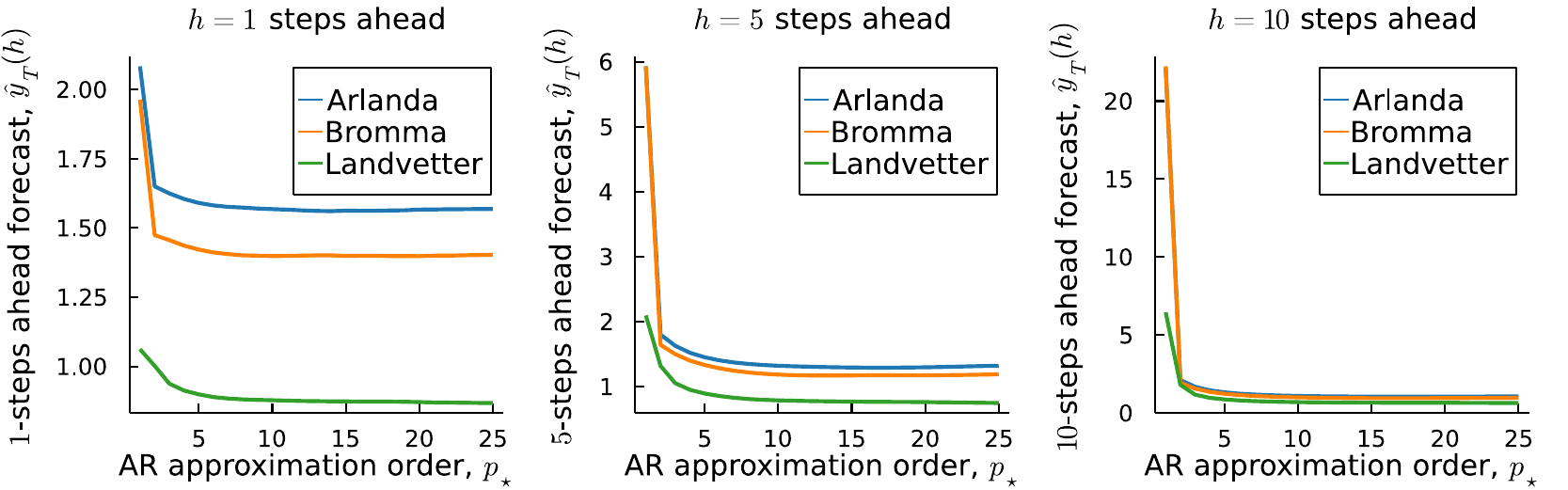}
    \caption{Investigating the convergence of forecasts as $p_\star$ increases in the $\mathrm{VARMA}(p_\star,0)$ approximation of the $\mathrm{VARTFIMA}(2,0)$ model. The forecasts are all produced standing at the last time point in the training data.}\label{fig:SwedTempForecastApprox}
\end{figure}

Figure \ref{fig:SwedTempRMSE} shows that the multivariate VARTFIMA(2,0) produces more accurate forecasts at all horizons than univariate ARTFIMA(2,0) models fitted to each time series separately, particularly at the longer horizons; note that a forecast horizon with $h=50$ corresponds to $10$ hours ahead. 

\begin{figure}
    \includegraphics[width=0.9\linewidth]{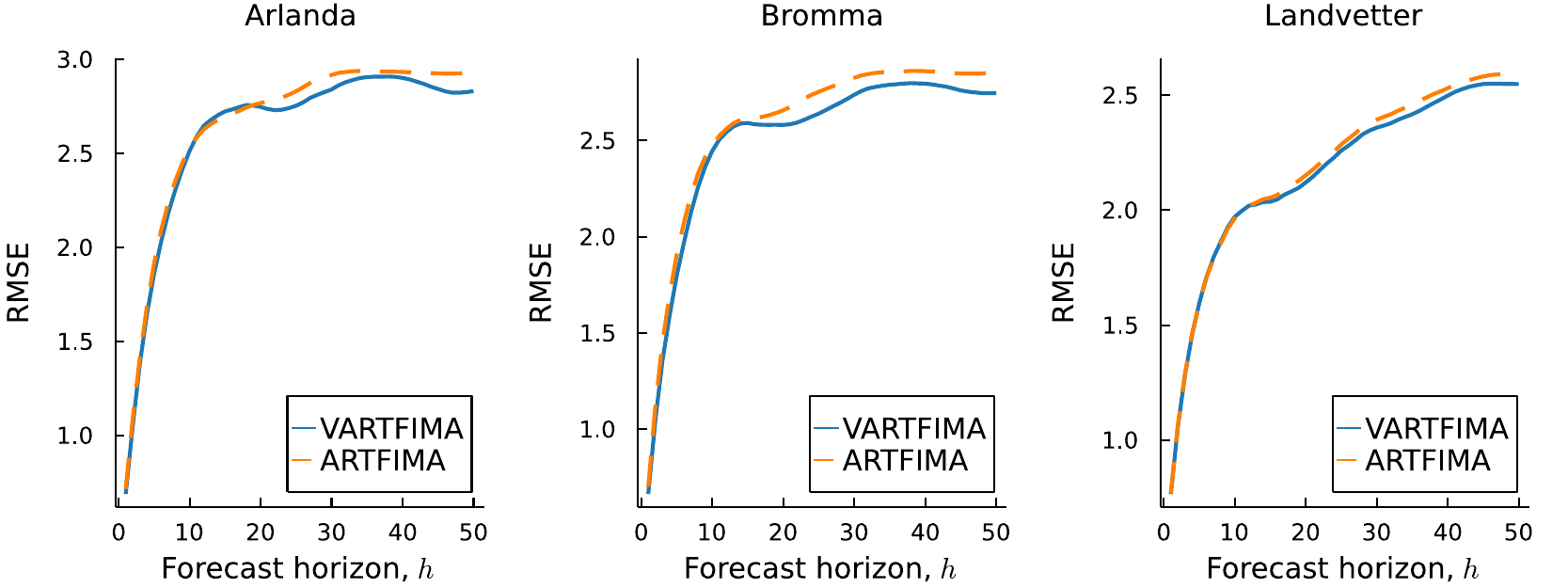}
    \caption{Comparing the out-of-sample RMSE forecasting performance of the VARTFIMA(2,0) model (solid blue) to univariate ARTFIMA(2,0) models fitted to each series (dashed orange) in the Swedish temperature data. Both models were fitted to the training data (Jan 3, 2016 - to Dec 21, 2018) and the estimates were then kept fixed throughout the forecasting period (Dec 21, 2018 - Dec 29, 2018).}\label{fig:SwedTempRMSE}
\end{figure}

\section{Conclusions}\label{sec:Conclusions}
Our paper proposes a subsampling MCMC approach for stationary multivariate time series models. Using a measure which takes into account both computing cost and the statistical inefficiency of likelihood estimators, we demonstrate a speed-up factor of up to two orders of magnitude on three datasets compared to MCMC using the full dataset. 

To test the proposed spectral subsampling MCMC in challenging problems, we propose a new multivariate time series model by extending the univariate ARTFIMA model to a multivariate setting. Some properties of this new model are derived, including its spectral density matrix, and our results show that the vector ARTFIMA model outperforms VARMA models in all three datasets. This suggest the VARTFIMA model as a useful model for modelling a wide range of multivariate time series and future work should investigate this more fully.

Our work demonstrates that MCMC sampling to explore the posterior based on the Whittle likelihood scales well to very long multivariate time series, and that the approximation to the exact time domain posterior is excellent. We further show that spectral subsampling can give an additional speed-up of two orders of magnitude on large time series in rather complex multivariate time series models with semi-long range dependence. The biggest challenge is to obtain good control variates in high-dimensional parameter spaces for models with potentially very non-Gaussian likelihoods. We are unable to obtain good enough control variates for the best fitting models on the shortest of the three dataset. This is probably the result of the VARMA and VARTFIMA class of models having challenging likelihood functions, especially the moving average (MA) part of the model. We demonstrate that longer time series are more amenable to subsampling than shorter ones, so spectral subsampling MCMC is a method most suitable to large scale problems, at least when it comes to VARMA/VARTFIMA models. Future research will extend the methodology to high-dimensional complex models by developing better control variates in high-dimensions that consider the local nonidentification aspects, and exploring the potential of alternative inference algorithms such as variational inference \citep{blei2017variational}. Variational inference with an estimated likelihood is less sensitive to the variability in the estimator \citep{tran2017variational}, and is thus an appealing alternative to explore.

\bibliographystyle{apalike}
\addcontentsline{toc}{section}{\refname}\bibliography{MultiSpectralMCMCSecondRev}
\appendix

\section{Proofs}\label{app:proofs}
\begin{proof}[Proof of Theorem \ref{thm:stationary}]
    Theorem 11.3.1 of \cite{brockwell1991time} shows that the condition $\Phi(z)$ for $|z| \leq 1$ implies that there exists $\epsilon>0$ such that the matrix $\Phi^{-1}(z)$ exists and has a power series expansion
    \begin{equation*}
        \Phi^{-1}(z) = \sum_{j=0}^\infty \mathbf{A}_j z^j, \text{ for } |z|<1+\epsilon, 
    \end{equation*}
    with the elements of $\mathbf{A}_j$ being absolutely summable.
    Inverting the operator $\Delta^{\mathbf{d},\v \lambda}$ gives
    \begin{equation*}
        \Delta^{\mathbf{-d},\v \lambda}\mathbf{Y}_t = \Big(\Delta^{-d_1,\lambda_1}Y_{1,t},\ldots,\Delta^{-d_r,\lambda_r}Y_{r,t}\Big)^\top,
    \end{equation*}
    where the inverted univariate operator is defined by \citep{Sabzikar2019}
    \begin{equation*}
        \Delta^{-d,\lambda}Y_t = (1-e^{-\lambda}L)^{-d}Y_t = \sum_{j=0}^\infty c_j^{-d,\lambda}Y_{t-j},
    \end{equation*}
    with $c_j^{-d,\lambda}=(-1)^j e^{-\lambda j}\binom{-d}{j}$. 
    Let $\Delta^{-\mathbf{d},\v \lambda}(z)$ be the $r \times r$ diagonal matrix with $k$th diagonal element equal to $(1-e^{-\lambda_k}z)^{-d_k}$.
    Direct matrix multiplication shows that the element in row $k$, column $l$ of the matrix $\Delta^{\mathbf{-d},\v \lambda}(z)\Phi^{-1}(z)$ is the product
    \begin{equation*}
        \Big(\sum_{i=0}^\infty c_i^{-d_k,\lambda_k} z^i \Big)
        \Big(\sum_{j=0}^\infty a_j^{(kl)} z^j \Big) = \sum_{v=0}^\infty b_v^{(kl)} z^v,  
    \end{equation*}
    where $a_j^{(kl)}$ is the element in row $k$, column $l$ of  $\mathbf{A}_j$ and $b_v^{(kl)} = \sum_{s=0}^v c_s^{-d_k,\lambda_k} a_{v-s}^{(kl)}$. \cite{Sabzikar2019}[Proof of Theorem 2.2a] show that the sequence $c_j^{-d,\lambda}$ is absolutely summable if $\lambda >0$ and $d \notin \mathbb{Z}$, and it is shown above that $\sum_{j=0}^\infty \vert a_j^{(kl)}\vert<\infty$. Since the product of two absolutely summable series is absolutely summable \citep[Ch. 4.17]{knopp1990theory}, $\sum_{v=0}^\infty b_v^{(kl)}$ is also absolutely summable for all $k,l$.
    
    Hence, by Proposition 3.1.1 in \cite{brockwell1991time} we can apply the operator $\Delta^{\mathbf{-d},\v \lambda}\Phi^{-1}(L)$ to both sides of
    \begin{equation}
        \Phi(L)\Delta^{\mathbf{d},\v \lambda}(\mathbf{Y}_{t}-\v \mu)=\Theta(L) \v \varepsilon_{t}
    \end{equation}
    to obtain
    \begin{equation}\label{eq:causal}
        \mathbf{Y}_{t} = \v \mu +  \v\Psi(L)\v \varepsilon_{t},
    \end{equation}
    where 
    \begin{equation*}
        \v\Psi(z)\equiv \sum_{j=0}^\infty \v\Psi_j z^j = \Delta^{\mathbf{-d},\v \lambda}(z)\Phi^{-1}(z)\v\Theta(z),\hspace{0.5cm}\text{ for } |z|\leq 1,
    \end{equation*}
    and the $\{ \v\Psi_j\}_{j=0}^\infty$ are absolutely summable elementwise since $\v \Theta(z)$ is a polynomial of finite order $q$. By Proposition 3.1.2 in \cite{brockwell1991time} the VARTFIMA process is therefore causal and stationary.
    
    The causal representation in \eqref{eq:causal} shows that $\{ \v\Psi_j\}_{j=0}^\infty$ acts as a time invariant linear filter on the iid sequence $\{ \v\varepsilon_t\}_{t=0}^\infty$. The spectral density of $\mathbf{Y}_{t}$ therefore follows from \citet[Theorem 11.8.3]{brockwell1991time} and is
    \begin{align*}
        f_{\mathbf{Y}}(\omega) &=  \v\Psi(e^{-i\omega})f_{\v \varepsilon}(\omega) \v\Psi(e^{-i\omega})^H \\
        &=\frac{1}{2\pi} \Delta^{\mathbf{-d},\v \lambda}(e^{-i\omega}) \Phi^{-1}(e^{-i\omega})\Theta(e^{-i\omega}) \Sigma_{\varepsilon} \Theta^H(e^{-i\omega})  \Phi^{-H}(e^{-i\omega})  \Delta^{\mathbf{-d},\v \lambda}(e^{-i\omega})^H,
    \end{align*}
    where $f_{\v \varepsilon}(\omega)=\Sigma_{\varepsilon}$ is the spectral density matrix of the white noise process $\{ \v\varepsilon_t\}$ and $\Delta^{\mathbf{-d},\v \lambda}(z) = \mathrm{Diag}\big((1-e^{-\lambda_1}z)^{-d_1},\ldots,(1-e^{-\lambda_r }z)^{-d_r}\big).$
\end{proof}

\section{Raw data and additional empirical results}\label{app:additionalresults}

\begin{figure}[H]
    \includegraphics[width=0.32\linewidth]{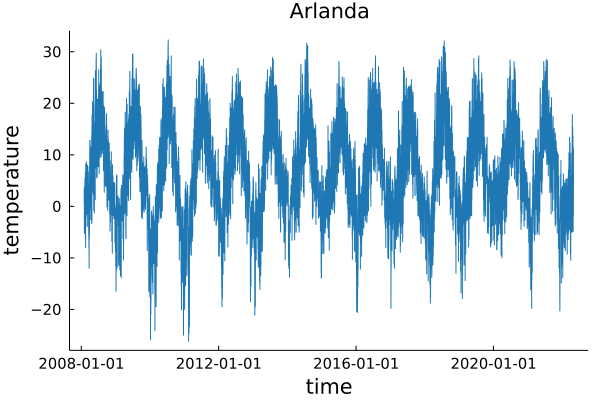}
    \includegraphics[width=0.32\linewidth]{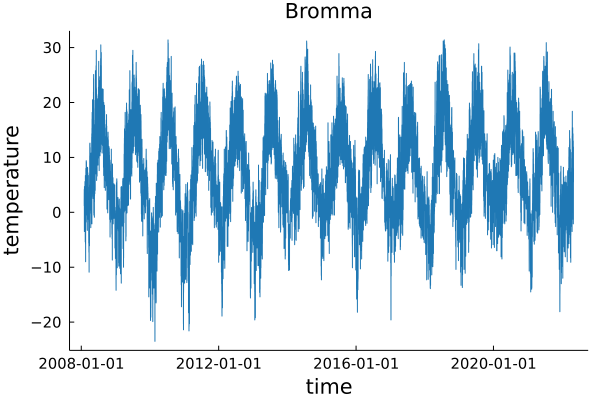}
    \includegraphics[width=0.32\linewidth]{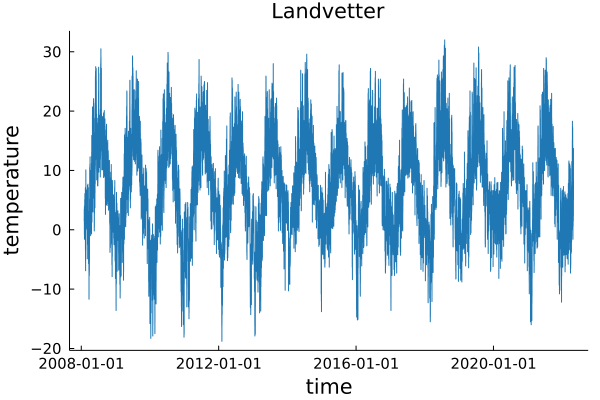}
    \caption{Swedish temperature data after interpolation, but before deseasoning.}\label{fig:SwedTempRawData}
\end{figure}

\begin{figure}[H]
    \includegraphics[width=0.35\linewidth]{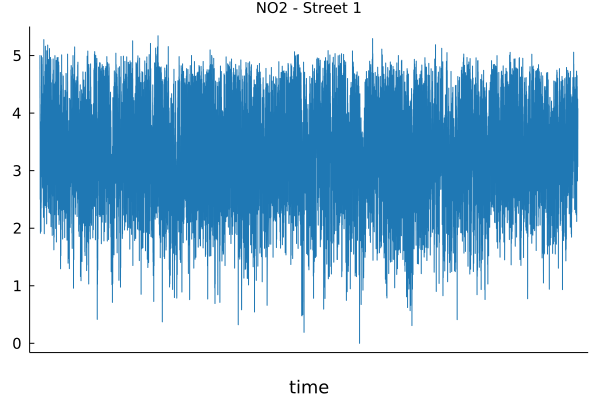}
    \includegraphics[width=0.35\linewidth]{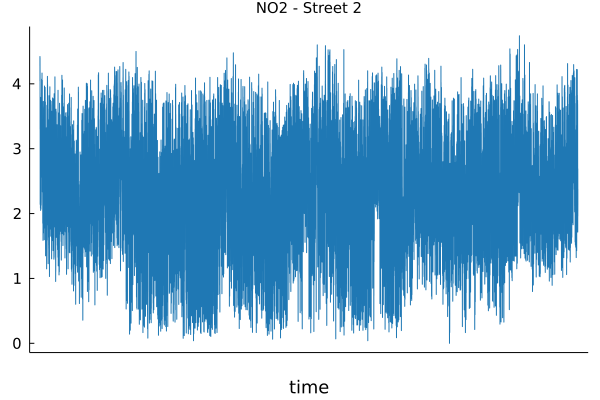} \\
    \includegraphics[width=0.35\linewidth]{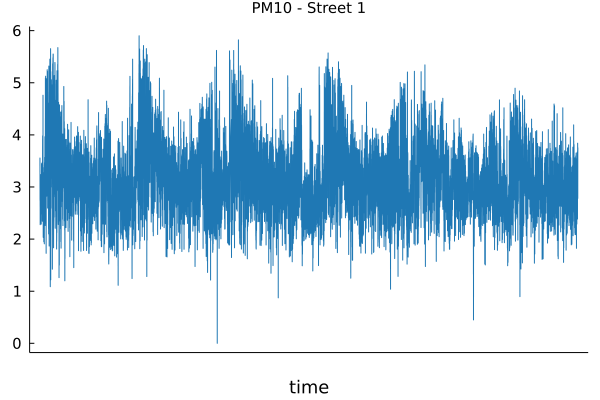}
    \includegraphics[width=0.35\linewidth]{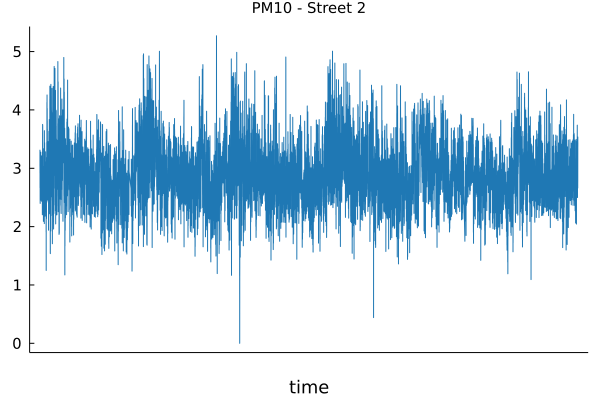}
    \caption{Stockholm air pollution data after interpolation and logarithmic transform, but before deseasoning.}\label{fig:SthlmPollutionRawData}
\end{figure}

\newpage

\begin{figure}
    \includegraphics[width=0.32\linewidth]{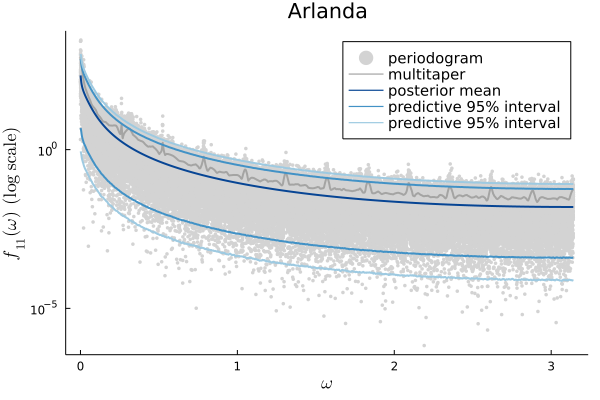}
    \includegraphics[width=0.32\linewidth]{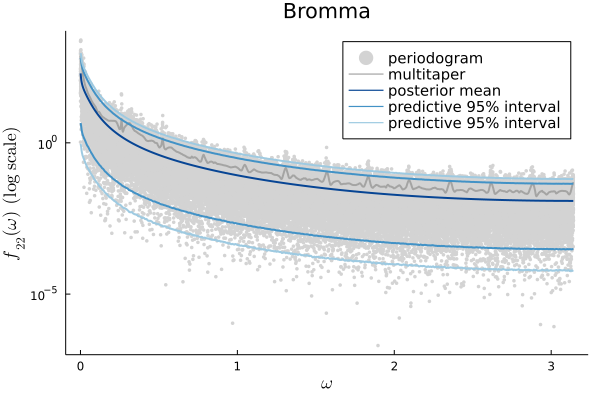}
    \includegraphics[width=0.32\linewidth]{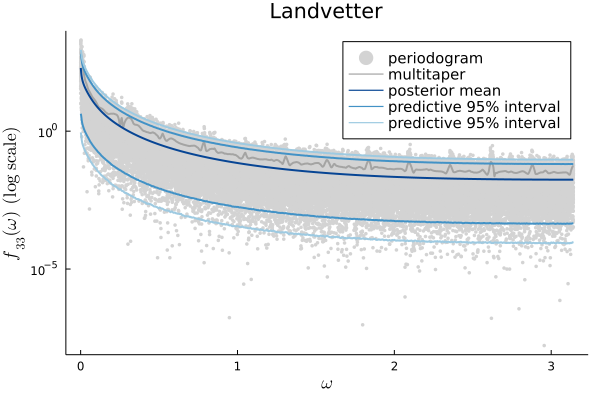}
    \caption{Posterior predictive fit of the univariate periodogram data for the VARTFIMA(2,0) model fitted to the Swedish temperature data. The predictive intervals are obtained by simulation from the asymptotic Whittle distribution $\mathcal{I}(\omega)\sim\mathrm{Expon}(f_{\v\theta}(\omega))$ with parameters $\v\theta$ drawn from the VARTFIMA(1,1) posterior.}\label{fig:SwedTempSpectral}
\end{figure}

\begin{figure}
    \includegraphics[width=0.4\linewidth]{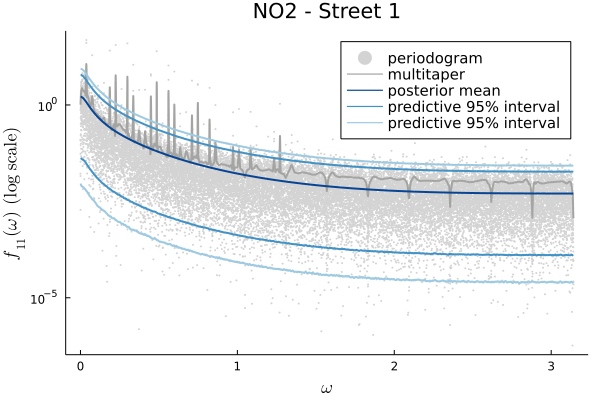}
    \includegraphics[width=0.4\linewidth]{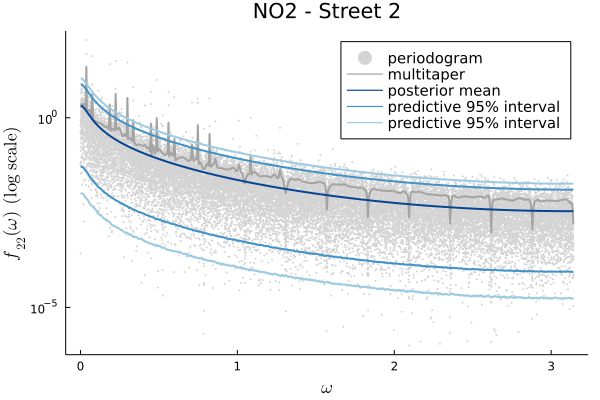} \\ 
    \includegraphics[width=0.4\linewidth]{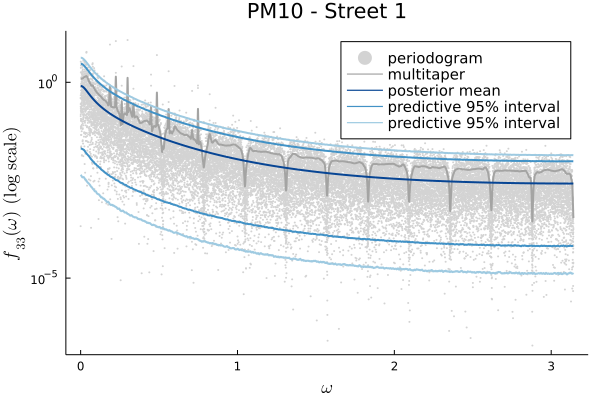}
    \includegraphics[width=0.4\linewidth]{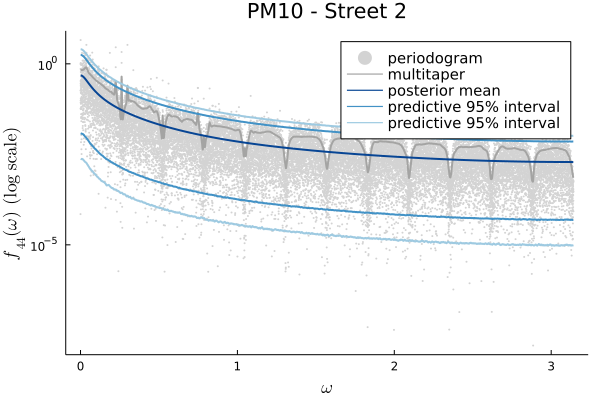}
    \caption{Posterior predictive fit of the univariate periodogram data for the VARTFIMA(2,0) model fitted to the Stockholm pollution data. The predictive intervals are obtained by simulation from the asymptotic Whittle distribution $\mathcal{I}(\omega)\sim\mathrm{Expon}(f_{\v\theta}(\omega))$ with parameters $\v\theta$ drawn from the VARTFIMA(1,1) posterior.}\label{fig:SthlmPollutionSpectral}
\end{figure}

\begin{figure}
    \includegraphics[width=0.80\linewidth]{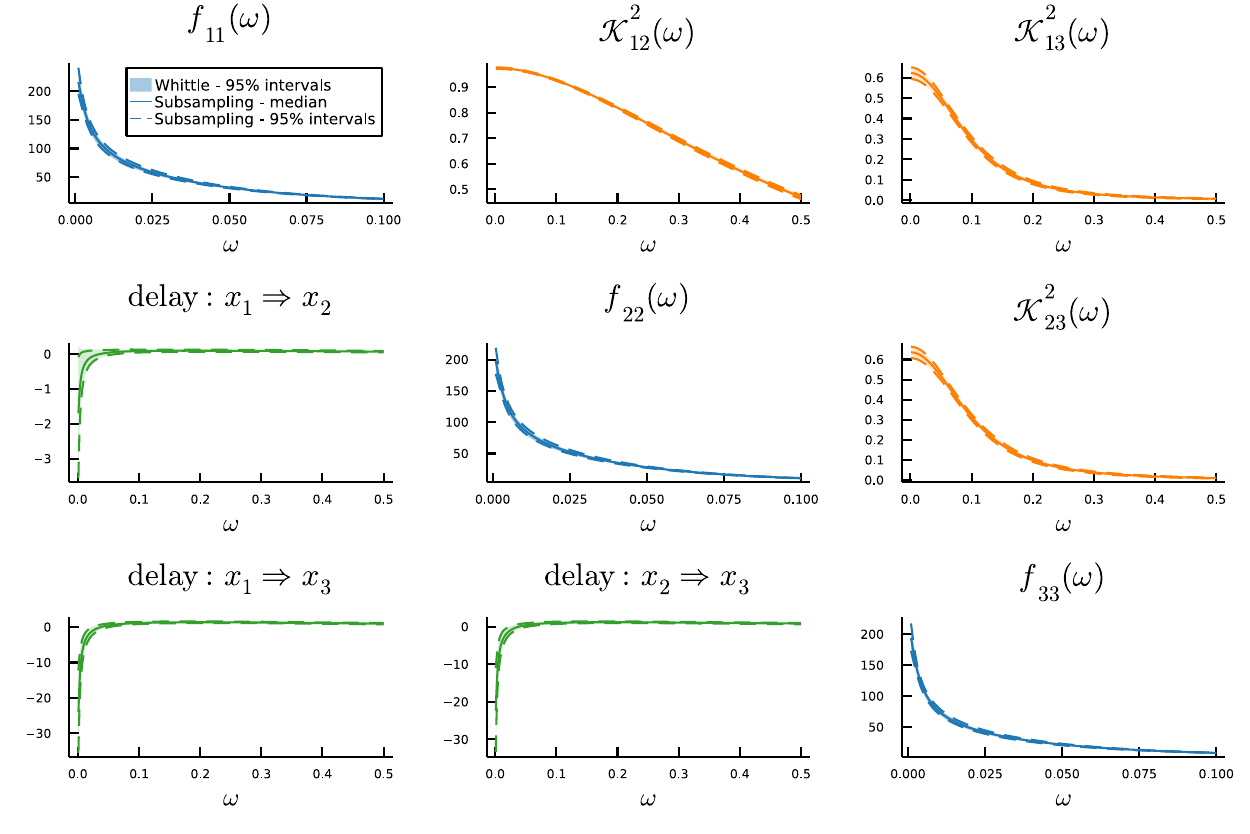}
    \caption{Posterior for the spectral density matrix in the VARTFIMA(2,0) model fitted to the Swedish temperature data. The plots on the diagonal are the marginal spectral densities for each time series. The plots above the diagonal are the squared coherences, and the plots below the diagonal are the time delays from the phase spectrum. The dashed lines in each subplot display the posterior median and 95\% credible intervals from spectral subsampling MCMC. The shaded regions are the 95\% credible intervals from the Whittle posterior on the whole dataset.}\label{fig:SwedTempCoherence}
\end{figure}

\begin{figure}
    \includegraphics[width=0.9\linewidth]{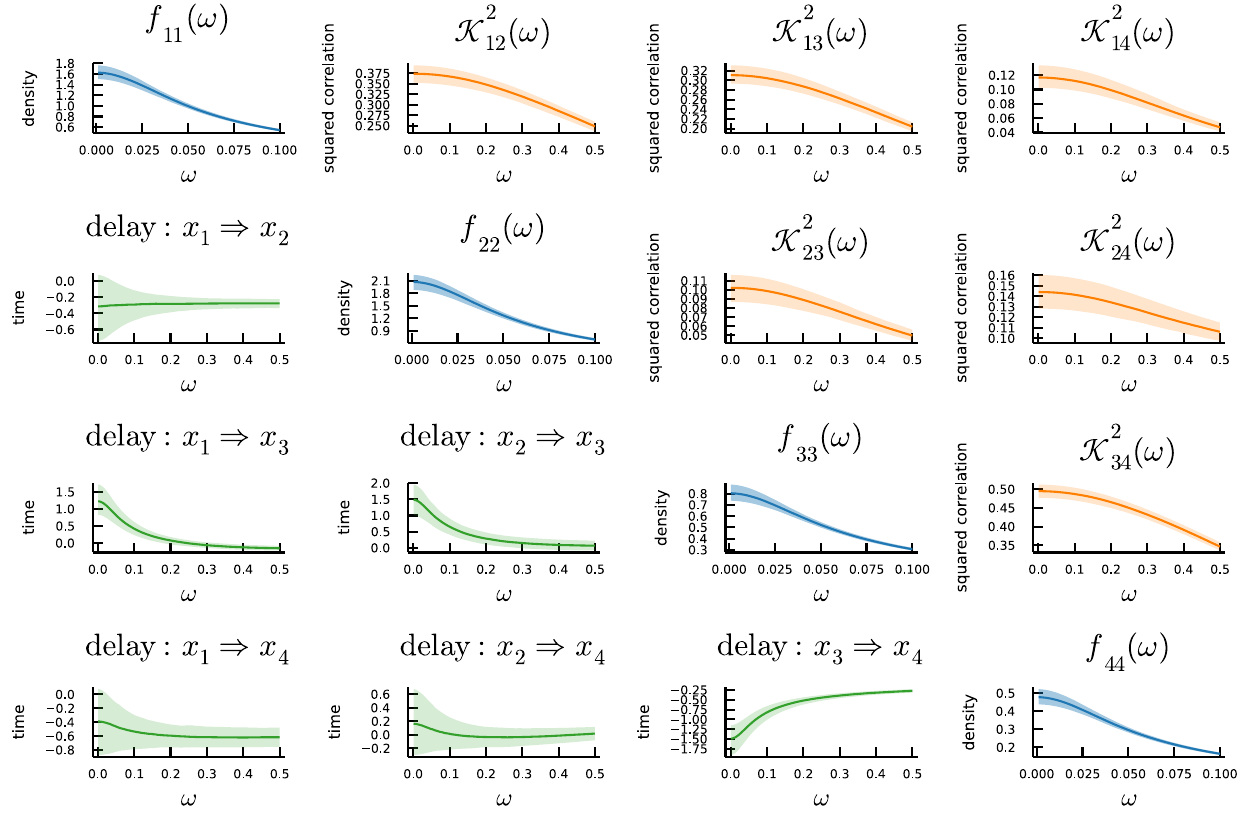}
    \caption{Posterior for the spectral density matrix in the VARTFIMA(2,0) model fitted to the Stockholm pollution data using the Whittle posterior on the whole dataset. The plots on the diagonal are the marginal spectral densities for each time series. The plots above the diagonal are the squared coherence, and the plots below the diagonal are the time delays from the phase spectrum. The posterior mean is marked out with a solid line and the shaded regions are the 95\% credible intervals.}\label{fig:SthlmPollutionCoherence}
\end{figure}

\begin{figure}
    \includegraphics[width=0.95\linewidth]{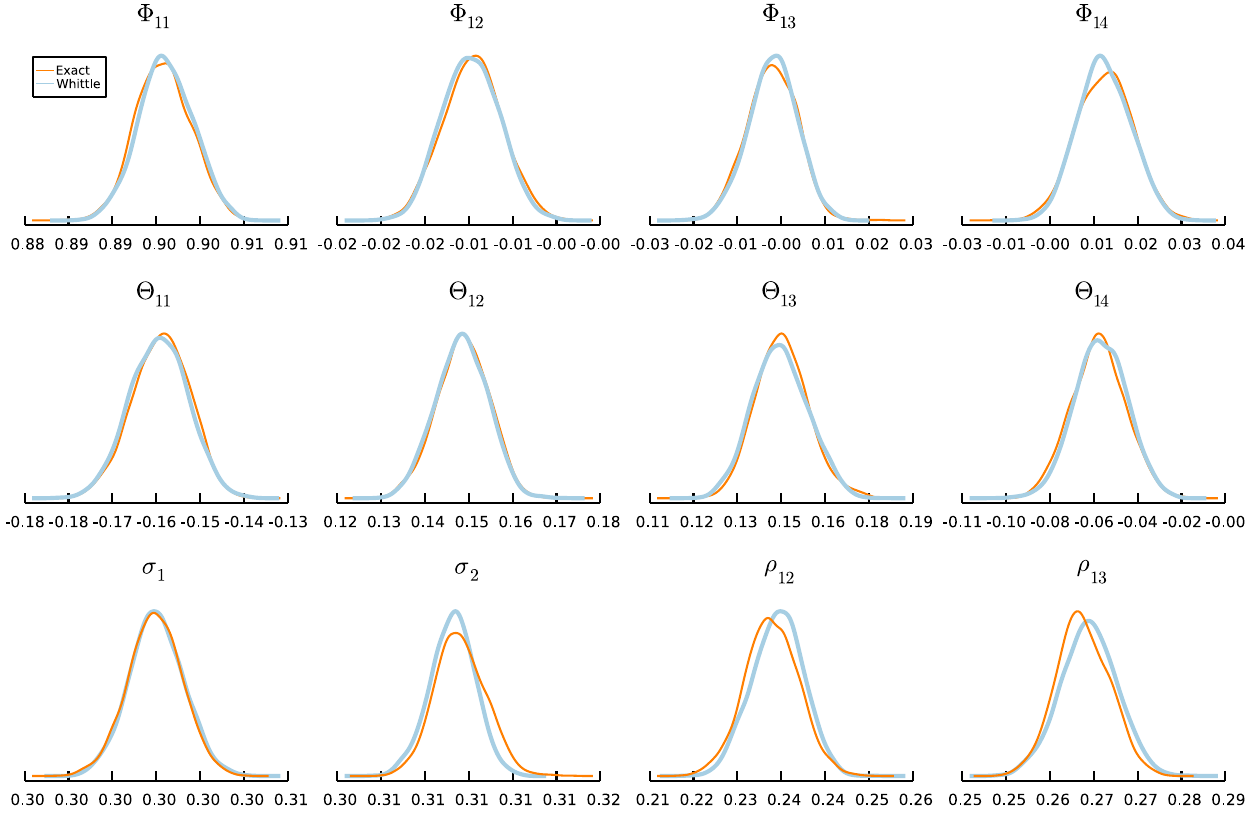}
    \caption{Kernel density estimates of a subset of the marginal posterior densities for the VARMA(1,1) model fitted to the Stockholm pollution data. The solid orange densities are from MCMC on the exact time domain posterior and the solid blue densities are from MCMC using the Whittle posterior on the whole dataset.}\label{fig:SthlmPollutionKDEarma11}
\end{figure}

\end{document}